\documentclass[12pt, draftclsnofoot, onecolumn]{IEEEtran}
\usepackage{amssymb}
\usepackage{amsmath}
\usepackage{amsthm}
\usepackage{graphicx}
\usepackage{algorithm}
\usepackage{algorithmic}
\usepackage{cite}
\usepackage{bm}
\usepackage{float}
\usepackage{multirow}
\usepackage{color}
\usepackage{subfigure}
\usepackage{caption}
\usepackage{epstopdf}
\usepackage{hyperref}
\usepackage{fixltx2e}

\theoremstyle{definition}

\theoremstyle{remark}
\newtheorem{remark}{Remark}

\theoremstyle{plain}

\newtheorem{lemma}{Lemma}
\newtheorem{proposition}{Proposition}
\newtheorem{corollary}{Corollary}

\begin{document}

\title{Non-Orthogonal Multiple Access (NOMA) With Multiple Intelligent Reflecting Surfaces}

\author{Yanyu~Cheng,~\IEEEmembership{Member,~IEEE},
        Kwok~Hung~Li,~\IEEEmembership{Senior Member,~IEEE},
        Yuanwei~Liu,~\IEEEmembership{Senior Member,~IEEE},
        Kah~Chan~Teh,~\IEEEmembership{Senior Member,~IEEE},
        and~George K.~Karagiannidis,~\IEEEmembership{Fellow,~IEEE}

\thanks{
Yanyu~Cheng, Kwok~Hung~Li, and Kah~Chan~Teh are with the School of Electrical and Electronic Engineering, Nanyang Technological University, Singapore 639798 (e-mail: ycheng022@e.ntu.edu.sg; ekhli@ntu.edu.sg; ekcteh@ntu.edu.sg).

Yuanwei~Liu is with the School of Electronic Engineering and Computer Science, Queen Mary University of London, London E1 4NS, U.K. (e-mail: yuanwei.liu@qmul.ac.uk).

George K.~Karagiannidis is with the Department of Electrical and Computer Engineering, Aristotle University of Thessaloniki, 54636 Thessaloniki, Greece (e-mail: geokarag@auth.gr).
}
}

\maketitle

\begin{abstract}
In this paper, non-orthogonal multiple access (NOMA) networks assisted by multiple intelligent reflecting surfaces (IRSs) with discrete phase shifts are investigated, in which each user device (UD) is served by an IRS to improve the quality of the received signal.
Two scenarios are considered according to whether there is a direct link between the base station (BS) and each UD, and the outage performance is analyzed for each of them.
Specifically, the asymptotic expressions for the upper and lower bounds of the outage probability in the high signal-to-noise ratio (SNR) regime are derived.
Following that, the diversity order is obtained.
It is shown that the use of discrete phase shifts does not degrade diversity order.
More importantly, simulation results reveal that a $3$-bit resolution for discrete phase shifts is sufficient to achieve near-optimal outage performance.
Simulation results also imply the superiority of IRSs over full-duplex decode-and-forward relays.
\end{abstract}

\vspace{0.3cm}


\IEEEpeerreviewmaketitle
\section{Introduction}\label{ch7_intro}
Non-orthogonal multiple access (NOMA) has been proposed as a candidate technique for future wireless networks \cite{ding2017application}.
The key idea of NOMA is to allocate multiple users to an orthogonal resource block, e.g., a time slot, a frequency band, or a spreading code, but with different power levels \cite{montalban2018multimedia,wang2019channel}.
It has been demonstrated that NOMA outperforms orthogonal multiple access (OMA) from the aspects of spectral efficiency, connection density, and user fairness \cite{liu2017non5g}.

On the other hand, intelligent reflecting surfaces (IRSs) are envisioned to provide reconfigurable wireless environments for future communication networks, which are also named reconfigurable intelligent surfaces (RISs) and large intelligent surfaces (LISs) \cite{liu2020reconfigurable}.
Specifically, an IRS consists of a large number of reconfigurable passive elements, and each element can induce a change of amplitude and phase for the incident signal \cite{mu2020exploiting,shen2019secrecy}.
By appropriately adjusting amplitude-reflection coefficients and phase-shift variables, it can improve link quality and enhance service coverage significantly \cite{zhou2020framework,dong2020secure,feng2020deep}.
Compared with the conventional communication assisting techniques, such as relays, IRSs consume less energy due to passive reflection and are able to operate in full-duplex (FD) mode without self-interference \cite{qingqing2019towards}.
Therefore, IRSs have been proposed as a cost-effective solution to enhance the spectral and energy efficiency of future wireless communication networks \cite{liu2017enhancing,cui2019secure}.
Also, IRSs have been introduced to NOMA networks for performance improvement.

\subsection{Related Work}
\textit{NOMA:} NOMA has been widely studied due to its potential application prospects.
In \cite{ding2014performance}, the authors evaluated the performance of a downlink NOMA system with randomly roaming users and showed that NOMA has better performance than the conventional multiple access (MA) techniques on ergodic sum rate.
Following that, they studied the impact of user pairing on the sum rate of fixed-power-allocation NOMA and cognitive-radio-inspired NOMA systems in \cite{ding2016impact}.
In \cite{oviedo2017fair}, the authors demonstrated that NOMA using a fair power allocation approach can always outperform OMA in terms of capacity.
A new evaluation criterion was proposed to analyze the performance gain of NOMA over OMA in \cite{xu2015new}, which considered not only sum rate but also individual rates.
In \cite{liu2016cooperative}, the authors applied the simultaneous wireless information and power transfer (SWIPT) to NOMA networks and proved that the use of SWIPT does not degrade diversity order as compared with the conventional NOMA.

\textit{IRSs with continuous phase shifts:} IRSs have attracted intensive research interest from both academia and industry \cite{zuo2020resource,ozdogan2019intelligent}.
For the application of IRSs, it is crucial to precisely estimate the channel state information (CSI).
A channel estimation protocol was proposed for IRS-assisted multiple-input single-output (MISO) networks in \cite{mishra2019channel}.
To shorten the training sequence, a channel estimation method was designed for downlink IRS-aided multiple-input multiple-output (MIMO) networks in \cite{he2019cascaded}.
A message-passing-based algorithm was proposed for uplink IRS-aided MIMO networks in \cite{liu2020matrix}.
IRSs have been demonstrated to outperform various types of relays, including both decode-and-forward (DF) and amplify-and-forward (AF) relays under half-duplex (HD) and FD modes.
For HD relaying networks, the data rate is scaled down by a factor of two, since two phases are required for the data transmission by the transmitter and relay \cite{di2020reconfigurable}.
For FD relaying networks, there is residual loop-back self-interference, which consequently degrades system performance \cite{di2020reconfigurable}.
Under the AF relaying protocol, when a relay node amplifies the received signal, it simultaneously amplifies the noise, while IRSs do not have that disadvantage \cite{qingqing2019towards,di2020reconfigurable}.
In \cite{bjornson2019intelligent}, IRSs were indicated to outperform HD DF relays when the number of reflecting elements is large enough.
In \cite{lyu2020spatial}, it was revealed that IRS-aided systems outperform FD DF relay (FDR)-aided systems when the number of IRSs exceeds a certain value.
Passive-beamforming design for IRSs is critical to system performance.
In \cite{zhang2019analysis}, multiple IRSs were deployed in a single-input single-output (SISO) system, and the passive beamforming of IRSs was optimized by minimizing outage probability.
Active and passive beamforming was jointly optimized for IRS-aided MISO systems in \cite{yu2019miso} and \cite{wu2019intelligent} with different objectives.
For physical layer security, IRSs can enhance system performance by blocking the signals of eavesdroppers \cite{han2020intelligent,guan2020intelligent}.

\textit{IRSs with discrete phase shifts:} The aforementioned works assume that IRSs have continuous phase shifts.
However, in practice, it is difficult to realize continuous phase shifts due to hardware limitations. Moreover, the complexity and cost will be extremely high, especially when there are a large number of reflecting elements.
Thus, it is more practical to consider discrete phase shifts.
In \cite{xu2019discrete}, the performance of an IRS-aided SISO system was evaluated from the data rate perspective.
In \cite{you2020channel}, the uplink transmission for an IRS-aided system was studied, and the problems of channel estimation and passive beamforming were investigated.
In \cite{guo2019weighted}, an IRS-aided MISO system was considered, and the active and passive beamforming were jointly optimized to maximize the weighted sum rate.
For IRS-aided MISO systems, the transmit power was minimized by designing active and passive beamforming in \cite{abeywickrama2020intelligent,wu2019beamformingICASSP,wu2019beamforming}.
The symbol error rate was minimized by optimizing active and passive beamforming in \cite{ye2020joint}.

\textit{IRS-assisted NOMA:} The combination of these two critical techniques, i.e., IRS and NOMA, has also been widely investigated.
In \cite{ding2020simple}, an IRS was deployed in a NOMA system to improve the coverage by assisting a cell-edge user device (UD) in data transmission.
In \cite{ding2020impact}, the impact of random phase shifting and coherent phase shifting for IRS-aided NOMA systems was further investigated.
In \cite{zhu2019power} and \cite{fu2019intelligent}, the beamforming vectors of the base station (BS) and IRS were optimized for IRS-assisted NOMA systems.
The aforementioned works assumed that the BS-IRS-UD channel is non-line-of-sight (NLoS). Since IRSs can be pre-deployed, the paths between the BS and IRSs may have line-of-sight (LoS) \cite{yang2020intelligent,liu2020ris,hou2019reconfigurable}.
For IRS-assisted NOMA networks, the authors assumed the BS-IRS-UD link to be LoS and optimized active and passive beamforming vectors in \cite{yang2020intelligent}.
The authors adopted machine learning approaches to jointly design power allocation, passive beamforming, and the position of IRS in \cite{liu2020ris}.
The performance of IRS-aided NOMA systems was analyzed in \cite{hou2019reconfigurable, cheng2020downlink}, where continuous phase shifts were assumed.

\subsection{Motivation and Contributions}
Most of the current studies focus on the centralized IRS, which has practical limitations \cite{gao2020distributed}. Alternatively, the distributed IRSs have been proposed to be more practical.
Moreover, the current works related to IRSs with discrete phase shifts mainly focus on optimizing beamforming \cite{xu2019discrete,guo2019weighted,you2020channel,abeywickrama2020intelligent,wu2019beamformingICASSP,wu2019beamforming,ye2020joint}, but lack corresponding performance analysis.
Due to taking into account IRSs with discrete phase shifts, it is a challenging task to evaluate system performance.
This motivates us to characterize the performance of NOMA systems assisted by multiple IRSs with discrete phase shifts, which is beneficial to future deployment of IRSs.
The contributions are summarized as follows:
\begin{itemize}
\item We characterize the performance of NOMA networks assisted by multiple IRSs with discrete phase shifts. More specifically, we consider two scenarios based on whether there are direct links between the BS and UDs.
\item To be practical, we adopt the Nakagami-$m$ fading model for BS-UD and BS-IRS-UD links so that those links can be either LoS or NLoS. Correspondingly, for each scenario, we derive the probability density function (PDF) and the cumulative distribution function (CDF) of the ordered channel gain by utilizing the Laplace transform.
\item We derive the high signal-to-noise ratio (SNR) asymptotic expressions for the upper and lower bounds of the outage probability for each scenario.
\item To gain further insight, we obtain the diversity order for each scenario. We demonstrate that discrete phase shifts do not jeopardize diversity order as compared with continuous phase shifts. Simulation results further reveal that a $3$-bit resolution for phase shifts is sufficient to achieve near-optimal outage performance.
\end{itemize}

\subsection{Organization and Notation}
The remainder of this paper is organized as follows.
In Section \ref{ch7_system_model}, the system model of multi-IRS assisted NOMA is described.
The performance analysis of Scenario \uppercase\expandafter{\romannumeral1} is conducted in Section \ref{ch7_s1}, followed by the analysis of Scenario \uppercase\expandafter{\romannumeral2} in Section \ref{ch7_s2}.
Numerical and simulation results are presented in Section \ref{ch7_numerical}.
Finally, our conclusion is drawn in Section \ref{ch7_summary}.

In this paper, scalars are denoted by italic letters. Vectors and matrices are denoted by bold-face letters.
For a vector $\mathbf{v}$, $\mathbf{v}^T$ denotes the transpose of $\mathbf{v}$.
$\mathrm{diag}(\mathbf{v})$ denotes a diagonal matrix in which each diagonal element is the corresponding element in $\mathbf{v}$, respectively.
$\mathrm{arg}(\cdot)$ denotes the argument of a complex number.
$\mathbb{C}^{x\times y}$ denotes the space of $x \times y$ complex-valued matrices.
$\mathrm{P_r}(\cdot)$ denotes the probability.

\section{System Model}\label{ch7_system_model}

\begin{figure}
\begin{center}
\includegraphics[width=0.9\linewidth]{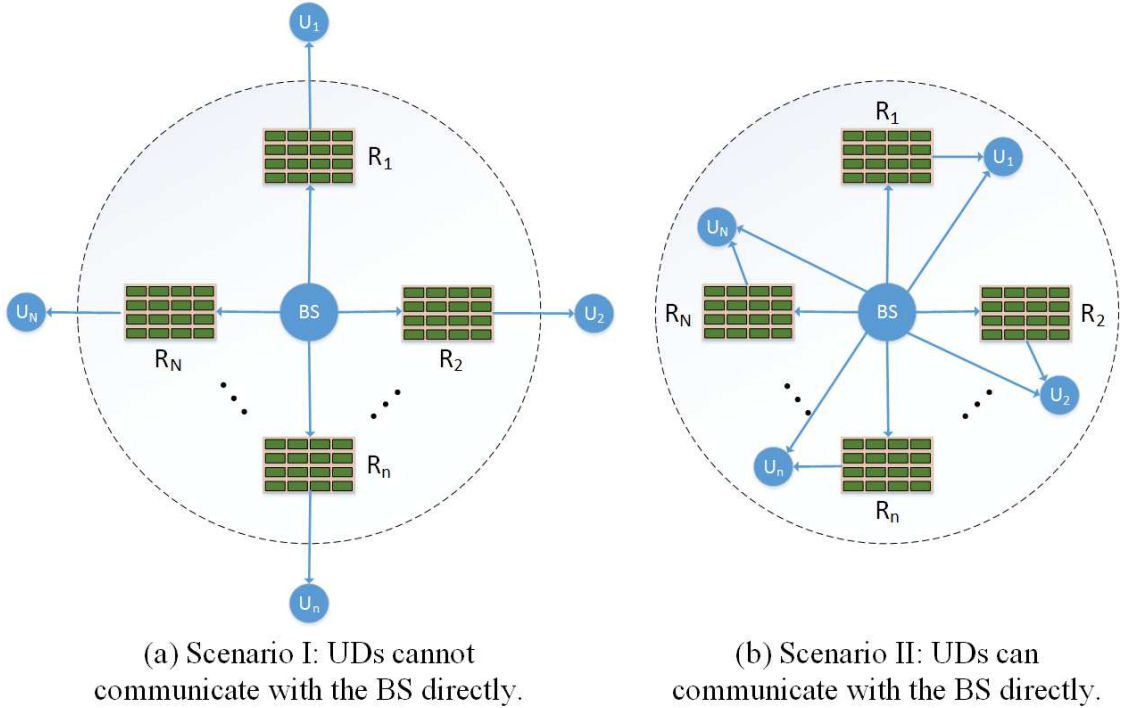}
\end{center}
\caption{The system models of multi-IRS assisted NOMA with discrete phase shifts for two scenarios.}
\label{ch7_Fig_system_model}
\end{figure}

We consider a NOMA network assisted by multiple IRSs with discrete phase shifts.
Specifically, there is a single-antenna BS and $N$ IRSs, denoted by R\textsubscript{$1$}, R\textsubscript{$2$}, $\cdots$, R\textsubscript{$N$}.
All IRSs can be pre-deployed and have a sufficient distance from each other.
Therefore, each IRS has a service coverage without overlapping with other IRSs' service coverages, and the interference signals reflected from neighboring IRSs can be neglected due to severe path loss \cite{ding2020simple}.
Within the service area of each IRS, we select a single-antenna UD for this IRS.
All $N$ selected UDs, denoted by U\textsubscript{$1$}, U\textsubscript{$2$}, $\cdots$, U\textsubscript{$N$}, form a NOMA group, and each UD is assisted by the corresponding IRS to communicate with the BS.
Based on the positions of all UDs, we consider two scenarios, Scenario \uppercase\expandafter{\romannumeral1} and Scenario \uppercase\expandafter{\romannumeral2}. More specifically, in Scenario \uppercase\expandafter{\romannumeral1}, all UDs are far away from the BS, and there is no direct link between the BS and each UD as shown in Fig. \ref{ch7_Fig_system_model}(a), while in Scenario \uppercase\expandafter{\romannumeral2}, all UDs can communicate with the BS directly as shown in Fig. \ref{ch7_Fig_system_model}(b).
Each IRS has $K$ discrete reflecting elements.
The reflection-coefficient matrix of R\textsubscript{$n$}, denoted by $\mathbf{\Theta}_n=\mathrm{diag}\left(\beta_1^n e^{j\theta_1^n}, \beta_2^n e^{j\theta_2^n}, \cdots, \beta_K^n e^{j\theta_K^n}\right)$  \big(note that $j=\sqrt{-1}$\big), can be intelligently adjusted, where $\beta_k^n \in [0,1]$ and $\theta_k^n$ are the amplitude-reflection coefficient and phase-shift variable of the $k$th element, respectively.
To be practical, we consider discrete phase shifts which have a $b$-bit resolution with $L=2^b$ levels.
By uniformly quantizing the interval $[0, 2\pi)$, we have the set of discrete phase-shift variables, $\mathcal{D}=\left\{\frac{\Delta}{2}, \frac{3\Delta}{2}, \cdots, \frac{(2L-1)\Delta}{2}\right\}$, where $\Delta=2\pi/L$.

\begin{remark}
In the system model, each IRS serves a UD in the NOMA group. It can also simultaneously serve other UDs that are not in this NOMA group within the service coverage.
Each IRS's parameters are optimized for the UD in the NOMA group, i.e., the prioritized UD.
For other UDs served by this IRS, it is a scenario of random phase shifts, and the corresponding performance has been analyzed \cite{ding2020impact}.
Hence, in this paper, we only analyze the performance of prioritized UDs.
Moreover, since IRSs can significantly improve channel quality, the performance of prioritized UDs can be guaranteed.
For each IRS, we can select the UD with the highest quality-of-service requirement to be the prioritized UD.
\end{remark}

\subsection{Channel Model}

For both scenarios, all channels experience quasi-static flat fading, and the CSI of all channels is assumed to be perfectly known at the BS.
As IRSs can be pre-deployed, both BS-IRS and IRS-UD links can be either LoS or NLoS. To realize this assumption, we adopt the Nakagami-$m$ fading model \cite{hou2019reconfigurable,cheng2020downlink}.
The channel between the BS and R\textsubscript{$n$} is denoted by $\mathbf{G}_n \in \mathbb{C}^{1\times K}$.
The channel between R\textsubscript{$n$} and U\textsubscript{$n$} is denoted by $\mathbf{g}_n \in \mathbb{C}^{K\times 1}$.
Particularly, they are $\mathbf{G}_n=[G_1^n, G_2^n, \cdots, G_K^n]$ and $\mathbf{g}_n=[g_1^n, g_2^n, \cdots, g_K^n]^T$, respectively.
All elements in $\mathbf{G}_n$ have a fading parameter of $m_G$, and all elements in $\mathbf{g}_n$ have a fading parameter of $m_g$.
For Scenario \uppercase\expandafter{\romannumeral2}, besides BS-IRS-UD links, there are direct links between the BS and UDs. The channel between the BS and U\textsubscript{$n$} is denoted by $h_n$, which also follows the Nakagami-$m$ fading model with a fading parameter of $m_h$.
In particular, it is NLoS for $m_\mathcal{G}=1$ and it is LoS for $m_\mathcal{G}>1$ $\left(\mathcal{G} \in \{G, g, h\}\right)$.

\subsection{Signal Model}
\subsubsection{Scenario \uppercase\expandafter{\romannumeral1}}
In Scenario \uppercase\expandafter{\romannumeral1}, there is no direct link between the BS and each UD.
Without loss of generality, assume that $|\mathbf{G}_1 \mathbf{\Theta}_1 \mathbf{g}_1| \le |\mathbf{G}_2 \mathbf{\Theta}_2 \mathbf{g}_2| \le \cdots \le |\mathbf{G}_N \mathbf{\Theta}_N \mathbf{g}_N|$.
The BS broadcasts the signal $x$ which is given by
\begin{equation}
\begin{split}
x=\sum_{n=1}^N \sqrt{\alpha_n P}s_n,
\end{split}
\end{equation}
where $P$ is the transmit power, $s_n$ is the signal transmitted to U\textsubscript{$n$}, and $\alpha_n$ is the power allocation coefficient of U\textsubscript{$n$} $\left(\sum_{n=1}^N \alpha_n=1\right)$.
Note that we assume the fixed power allocation sharing between all UDs and set $\alpha_1 > \alpha_2 > \cdots > \alpha_N$ for user fairness \cite{ding2020simple,ding2020impact,hou2019reconfigurable}.
The received signal at U\textsubscript{$n$} is expressed as
\begin{equation}
\begin{split}
y_n^1=\mathbf{G}_n \mathbf{\Theta}_n \mathbf{g}_n x+n_n,
\end{split}
\end{equation}
where $n_n$ is the additive white Gaussian noise (AWGN) at U\textsubscript{$n$}.
Based on the successive interference cancellation (SIC) principle, U\textsubscript{$n$} needs to detect the signals of U\textsubscript{$l$} $(l \le n)$ in a successive manner.
The signal-to-interference-plus-noise ratio (SINR) of detecting U\textsubscript{$l$}'s signal at U\textsubscript{$n$} $(l \le n < N\ \mathrm{or}\ l < n = N)$ is given by
\begin{equation}\label{ch7_eq_Without_SINR}
\begin{split}
\Psi_{n\leftarrow l}^1=\frac{|\mathbf{G}_n \mathbf{\Theta}_n \mathbf{g}_n|^2 \alpha_l}{|\mathbf{G}_n \mathbf{\Theta}_n \mathbf{g}_n|^2 \sum_{i=l+1}^N \alpha_i+ \frac{1}{\rho}},
\end{split}
\end{equation}
where $\rho=P/\sigma_n^2.$
As a special case, when $l = n = N$, the SNR of decoding U\textsubscript{$N$}'s signal at U\textsubscript{$N$} is given by
\begin{equation}
\begin{split}
\Psi_{N\leftarrow N}^1=|\mathbf{G}_N \mathbf{\Theta}_N \mathbf{g}_N|^2 \alpha_N\rho.
\end{split}
\end{equation}

\subsubsection{Scenario \uppercase\expandafter{\romannumeral2}}
As compared with Scenario \uppercase\expandafter{\romannumeral1}, the difference is that there are direct links between the BS and all UDs in Scenario \uppercase\expandafter{\romannumeral2}.
Without loss of generality, we assume that $|h_1 + \mathbf{G}_1 \mathbf{\Theta}_1 \mathbf{g}_1| \le |h_2 + \mathbf{G}_2 \mathbf{\Theta}_2 \mathbf{g}_2| \le \cdots \le |h_N + \mathbf{G}_N \mathbf{\Theta}_N \mathbf{g}_N|$.
The received signal at U\textsubscript{$n$} is given by
\begin{equation}
\begin{split}
y_n^2=(h_n + \mathbf{G}_n \mathbf{\Theta}_n \mathbf{g}_n) x+n_n.
\end{split}
\end{equation}
The SINR of detecting U\textsubscript{$l$}'s signal at U\textsubscript{$n$} $(l \le n < N\ \mathrm{or}\ l < n = N)$ is given by
\begin{equation}
\begin{split}
\Psi_{n\leftarrow l}^2=\frac{|h_n + \mathbf{G}_n \mathbf{\Theta}_n \mathbf{g}_n|^2 \alpha_l}{|h_n + \mathbf{G}_n \mathbf{\Theta}_n \mathbf{g}_n|^2 \sum_{i=l+1}^N \alpha_i+ \frac{1}{\rho}}.
\end{split}
\end{equation}
For the case of $l = n = N$, the SNR of decoding U\textsubscript{$N$}'s signal at U\textsubscript{$N$} is given by
\begin{equation}
\begin{split}
\Psi_{N\leftarrow N}^2=|h_n + \mathbf{G}_N \mathbf{\Theta}_N \mathbf{g}_N|^2 \alpha_N\rho.
\end{split}
\end{equation}

\section{Scenario \uppercase\expandafter{\romannumeral1} (Without Direct Link)}\label{ch7_s1}
In this section, we will analyze the outage performance of Scenario \uppercase\expandafter{\romannumeral1}.

\subsection{IRS Parameters for Scenario \uppercase\expandafter{\romannumeral1}}
We aim to provide the best channel quality to each UD by adjusting the parameters of each IRS.
For the BS-IRS-U\textsubscript{$n$} link, it is to maximize $|\mathbf{G}_n \mathbf{\Theta}_n \mathbf{g}_n| =\left|\sum_{k=1}^{K}\beta_k^n G_k^n g_k^n e^{j\theta_k^n}\right|$. This can be achieved by intelligently adjusting $\theta_k^n$, the phase-shift variable of each element.
The optimal discrete phase-shift variables are derived below.
\begin{lemma}
Denote the optimal discrete phase-shift variable by $\hat{\theta}_k^{1,n}$. Then, we have
\begin{equation}\label{ch7_eq_Lemma_Link_Without}
\begin{split}
\hat{\theta}_k^{1,n} = \Delta \left( \left\lfloor \frac{\bar{\theta}_k^{1,n}}{\Delta} \right\rfloor+\frac{1}{2}\right), \quad k = 1, 2, \cdots, K,
\end{split}
\end{equation}
where $\lfloor \cdot \rfloor$ is the floor function, $\bar{\theta}_k^{1,n}=\tilde{\theta}_n-\mathrm{arg}(G_k^n g_k^n)$, and $\tilde{\theta}_n$ is an arbitrary constant ranging in $[0, 2\pi)$.
\end{lemma}

\begin{proof}
By adjusting phase-shift variables to make all phases of $G_k^n g_k^n e^{j\theta_k^n}$ to be the same, we can obtain the optimal continuous phase-shift variables.
There is more than one solution for $\left\{\theta_k^n\right\}$, and the generalized solution is $\bar{\theta}_k^{1,n}$.
Since we consider discrete phase shifts, we can adjust the phase-shift variable of each element to be close to the optimal continuous phase-shift variable and obtain the optimal discrete phase-shift variable as \eqref{ch7_eq_Lemma_Link_Without}. This completes the proof.
\end{proof}

The equivalent channel gain after adopting optimal continuous phase shifts $\big\{\bar{\theta}_k^{1,n}\big\}$ is given by $|\mathbf{G}_n \mathbf{\Theta}_n \mathbf{g}_n| = \beta \left|\sum_{k=1}^{K} G_k^n g_k^n e^{j\bar{\theta}_k^{1,n}}\right|
= \beta\sum_{k=1}^{K} |G_k^n| |g_k^n|$, where we assume that $\beta_k^n=\beta$, $\forall n, k$, without loss of generality.
For discrete phase shifts, the quantization error at the $k$th element is given by $\theta_k^{1,n,er}=\hat{\theta}_k^{1,n}-\bar{\theta}_k^{1,n}$, which is uniformly distributed in $\big[-\frac{\Delta}{2}, \frac{\Delta}{2}\big)$.
As such, after adopting the optimal discrete phase-shift variables $\big\{\hat{\theta}_k^{1,n}\big\}$, we can obtain the equivalent channel gain as
\begin{equation}
\begin{split}
|\mathbf{G}_n \mathbf{\Theta}_n \mathbf{g}_n| & = \beta \left|\sum_{k=1}^{K} G_k^n g_k^n e^{j\hat{\theta}_k^{1,n}}\right|
=\beta \left|\sum_{k=1}^{K} G_k^n g_k^n e^{j\bar{\theta}_k^{1,n}}e^{j\theta_k^{1,n,er}}\right|
= \beta \left|\sum_{k=1}^{K} |G_k^n| |g_k^n| e^{j\theta_k^{1,n,er}}\right|.
\end{split}
\end{equation}

\begin{figure}
\begin{center}
\includegraphics[width=0.75\linewidth]{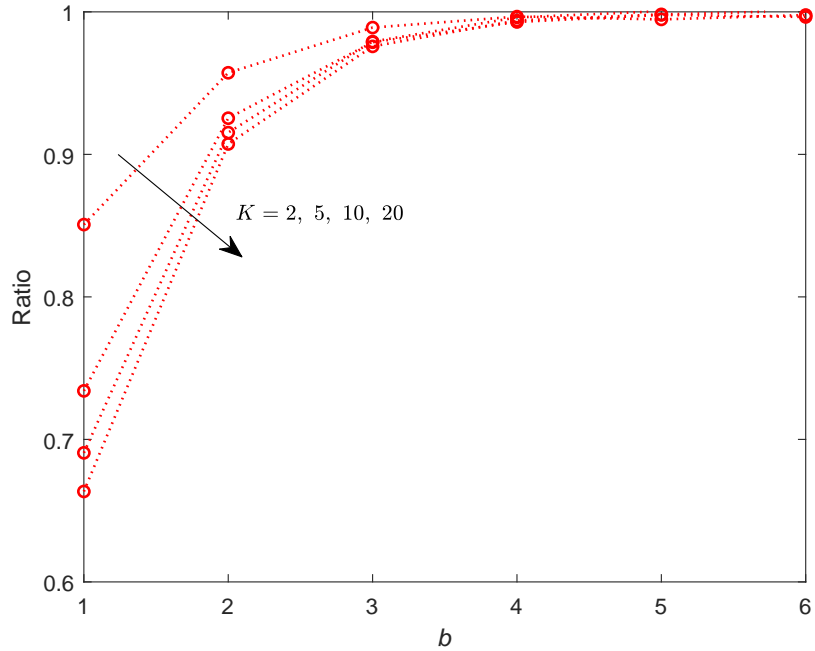}
\end{center}
\caption{The ratio of the average equivalent channel gain with discrete phase shifts to that with continuous phase shifts when $m_G=2$, $m_g=1$, and $\beta=0.9$ in Scenario \uppercase\expandafter{\romannumeral1}.}
\label{ch7_Fig_channel_gain_vs_b_without}
\end{figure}

To verify our analysis of the optimal discrete phase shifts, we have conducted Monte Carlo simulations, and the results are shown in Fig. \ref{ch7_Fig_channel_gain_vs_b_without}.
It is observed that the ratio of the equivalent channel gain with discrete phase shifts to that with continuous phase shifts increases as the value of $b$ increases, and gradually approaches $1$.

\subsection{Channel Statistics without Direct Link}
The derivation of channel statistics for IRS-aided systems is a challenging task.
In this paper, we take into account the discrete phase shifts, which further increases the difficulty of the derivation.
It is proved that the central limit theorem (CLT)-based channel statistics of the BS-IRS-UD link are inaccurate when the channel gain is near $0$ \cite{ding2020impact,cheng2020downlink}.
As we aim to obtain the diversity order, we cannot use the CLT to derive channel statistics.
On the other hand, the Laplace transform can be used to derive the accurate channel statistics for the channel gain near $0$ \cite{ding2020impact,cheng2020downlink}.
Therefore, we directly use it to derive new channel statistics of ordered UDs.
Consequently, we can gain some insightful results, such as diversity order, although the closed-form expression for the outage probability cannot be derived.
The channel statistics for Scenario \uppercase\expandafter{\romannumeral1} are shown in the following lemma.

\begin{lemma}\label{ch7_Lemm_Without}
Denote that $Y_n=\beta \left|\sum_{k=1}^{K} |G_k^n| |g_k^n| e^{j\theta_k^{1,n,er}}\right|$ which is the $n$th channel gain in ascending order. Based on order statistics, when $m_G\neq m_g$ and $b\ge 2$, the CDF of $Y_n$'s lower bound for $y\rightarrow0^+$ is given by
\begin{equation}\label{ch7_eq_Lemma_Without_1}
\begin{split}
F_{Y_n,low}^{0^+}(y) = &\frac{N!}{(N-n)!(n-1)!} \sum_{i=0}^{N-n} \binom{N-n}{i}\frac{(-1)^i}{n+i}
\xi_1^{n+i}y^{2m_sK(n+i)},
\end{split}
\end{equation}
where $a=\beta \cos\left(\frac{\pi}{2^b}\right)$, $m_s=\min\{m_G, m_g\}$, $m_l=\max\{m_G, m_g\}$, $\xi_1=\zeta_1/a^{2m_sK}$, $\zeta_1=\frac{\big(\sqrt{\pi}4^{m_s-m_l+1}(m_sm_l)^{m_s}\Gamma(2m_s)\Gamma(2m_l-2m_s)\big)^K}
{2m_sK \Gamma(2m_sK) \big(\Gamma(m_s)\Gamma(m_l)\Gamma\left(m_l-m_s+\frac{1}{2}\right)\big)^K}$, and $\Gamma(\cdot)$ is the gamma function.

Note that when $b\rightarrow \infty$, it is the case of continuous phase shifts. The CDF of $Y_n$ for $y\rightarrow0^+$ when $b\rightarrow \infty$ is given by
\begin{equation}\label{ch7_eq_Lemma_Without_2}
\begin{split}
F_{Y_n}^{0^+,b\rightarrow \infty}(y) = &\frac{N!}{(N-n)!(n-1)!}\sum_{i=0}^{N-n} \binom{N-n}{i}
\frac{(-1)^i}{n+i}\xi_2^{n+i}y^{2m_sK(n+i)},
\end{split}
\end{equation}
where $\xi_2=\zeta_1/\beta^{2m_sK}$.
\end{lemma}

\begin{proof}
See Appendix \ref{ch7_Appe_Lemm_Without}.
\end{proof}

\subsection{Performance Analysis of NOMA}

The outage probability of U\textsubscript{$n$} under the NOMA scheme is given by
\begin{equation}
\begin{split}
\mathbb{P}_n^1=1-\mathrm{P_r}(\Psi_{n\leftarrow 1}^1 \ge \tilde{\gamma}_1, \Psi_{n\leftarrow 2}^1 \ge \tilde{\gamma}_2, \cdots, \Psi_{n\leftarrow n}^1 \ge \tilde{\gamma}_n),
\end{split}
\end{equation}
where $\tilde{\gamma}_i=2^{\tilde{R}_i}-1$ with $\tilde{R}_i$ being the target rate of U\textsubscript{$i$}.
Furthermore, $\mathbb{P}_n^1$ can be transformed into
\begin{equation}
\begin{split}
\mathbb{P}_n^1&=1-\mathrm{P_r}\left(Y_n^2 \ge \tilde{\rho}_{n\leftarrow 1}, Y_n^2 \ge \tilde{\rho}_{n\leftarrow 2}, \cdots, Y_n^2 \ge \tilde{\rho}_{n\leftarrow n}\right)\\
&=\mathrm{P_r}\left(Y_n^2 < \tilde{\rho}_{n,max}\right),
\end{split}
\end{equation}
where $\tilde{\rho}_{n,max}=\max\left\{\tilde{\rho}_{n\leftarrow 1}, \tilde{\rho}_{n\leftarrow 2}, \cdots, \tilde{\rho}_{n\leftarrow n}\right\}$ and
\begin{equation}
\begin{split}
& \tilde{\rho}_{n\leftarrow l} = \left\{
\begin{aligned}
&\frac{\tilde{\gamma}_l}{\rho\left(\alpha_l-\tilde{\gamma}_l\sum_{i=l+1}^N \alpha_i\right)}, l \le n < N\ \mathrm{or}\ l < n = N,\\
&\frac{\tilde{\gamma}_N}{\rho\alpha_N}, \qquad \qquad \qquad \ \ \ \  l = n = N.
\end{aligned}
\right.
\end{split}
\end{equation}
Note that we need to ensure $\alpha_l-\tilde{\gamma}_l\sum_{i=l+1}^N \alpha_i>0$ for $1 \le l \le N-1$. This is because the ceiling of \eqref{ch7_eq_Without_SINR} is $\alpha_l/\big(\sum_{i=l+1}^N \alpha_i\big)$, which needs to be greater than $\tilde{\gamma}_l$ so that it is possible to implement the SIC successfully.

Based on Lemma \ref{ch7_Lemm_Without}, we can analyze the outage performance in the high-SNR regime as shown in the following proposition.
\begin{proposition}\label{ch7_Prop_Without_NOMA}
Under the NOMA scheme in Scenario \uppercase\expandafter{\romannumeral1}, when $b\ge 2$, the high-SNR asymptotic expressions for the upper and lower bounds of U\textsubscript{$n$}'s outage probability are as follows:
\begin{equation}\label{ch7_eq_Without_1}
\begin{split}
\mathbb{P}_n^{1, up, \infty} = \frac{N!}{(N-n)!n!}  \xi_1^n \tilde{\rho}_{n,max}^{nm_sK},
\end{split}
\end{equation}
\begin{equation}\label{ch7_eq_Without_3}
\begin{split}
\mathbb{P}_n^{1, low, \infty}&=\frac{N!}{(N-n)!n!} \xi_2^n \tilde{\rho}_{n,max}^{nm_sK}.
\end{split}
\end{equation}
Note that \eqref{ch7_eq_Without_3} is also the outage probability for the case of continuous shifts.
\end{proposition}

\begin{proof}
For the upper bound, we have
\begin{equation}
\begin{split}
\mathbb{P}_n^{1,up,\infty}=&F_{Y_n,low}^{0^+}\left(\sqrt{\tilde{\rho}_{n,max}}\right)
= \frac{N!}{(N-n)!(n-1)!}
\sum_{i=0}^{N-n} \binom{N-n}{i}\frac{(-1)^i}{n+i} \xi_1^{n+i} \tilde{\rho}_{n,max}^{m_sK(n+i)}.
\end{split}
\end{equation}
Then, by extracting the lowest-order term, we can obtain \eqref{ch7_eq_Without_1}.

It is the case of continuous shifts when $b \rightarrow \infty$, and the corresponding outage probability is the lower bound.
We have
\begin{equation}
\begin{split}
\mathbb{P}_n^{1,low,\infty}=&F_{Y_n}^{0^+,b \rightarrow \infty}\left(\sqrt{\tilde{\rho}_{n,max}}\right)
= \frac{N!}{(N-n)!(n-1)!}
\sum_{i=0}^{N-n} \binom{N-n}{i}\frac{(-1)^i}{n+i} \xi_2^{n+i} \tilde{\rho}_{n,max}^{m_sK(n+i)}.
\end{split}
\end{equation}
Next, after extracting the lowest-order term, we can obtain \eqref{ch7_eq_Without_3}.
This completes the proof.
\end{proof}

\begin{corollary}\label{ch7_Coro_Without_D}
In Scenario \uppercase\expandafter{\romannumeral1}, the diversity orders of the $n$th UD for both discrete and continuous phase shifts under the NOMA scheme are $nm_sK$.
\end{corollary}

\begin{proof}
Since $\tilde{\rho}_{n,max}\propto \rho^{-1}$, the diversity orders can be obtained based on \eqref{ch7_eq_Without_1} and \eqref{ch7_eq_Without_3}. This completes the proof.
\end{proof}

\begin{remark}
For Scenario \uppercase\expandafter{\romannumeral1}, the proposed system with $b$-bit discrete phase shifts achieves the same diversity order as that with continuous phase shifts when $b\ge2$, which implies that the application of discrete phase shifts does not jeopardize the diversity order.
The diversity order is related to the number of IRS reflecting elements, the Nakagami fading parameters of the BS-IRS-UD link, and the order of the channel gain.
\end{remark}

\subsection{Performance Analysis of OMA}
The comparison between NOMA and OMA has been well investigated in existing works, such as \cite{ding2014performance,ding2016impact,oviedo2017fair,xu2015new,liu2016cooperative}.
Hence, we provide the results of OMA as a special case of NOMA.
We assume that all UDs occupy a resource block under the OMA scheme for fair comparisons, and each UD is allocated $1/N$ resource block \cite{ding2014performance,oviedo2017fair,xu2015new}.
Since there is no ordering for the channel gains in the OMA scheme, all UDs have the same outage probability, which is given by
\begin{equation}
\begin{split}
\mathbb{P}_{OMA}^1=\mathrm{P_r}(\Psi_{OMA}^1 < \tilde{\gamma}_{OMA}),
\end{split}
\end{equation}
where $\Psi_{OMA}^1=\rho|\mathbf{G}_n \mathbf{\Theta}_n \mathbf{g}_n|^2$ and $\tilde{\gamma}_{OMA}=2^{N\tilde{R}}-1$ with $\tilde{R}$ being the target rate of each UD.
Furthermore, $\mathbb{P}_{OMA}^1$ can be transformed into
\begin{equation}
\begin{split}
\mathbb{P}_{OMA}^1&=\mathrm{P_r}\left(Y^2 < \frac{\tilde{\gamma}_{OMA}}{\rho}\right),
\end{split}
\end{equation}
where $Y$ is the unordered equivalent channel gain.

\begin{proposition}
Under the OMA scheme in Scenario \uppercase\expandafter{\romannumeral1}, when $b\ge 2$, the high-SNR asymptotic expressions for the upper and lower bounds of U\textsubscript{$n$}'s outage probability are as follows:
\begin{equation}\label{ch7_eq_Without_2}
\begin{split}
\mathbb{P}_{OMA}^{1,up,\infty}& = \xi_1 \tilde{\gamma}_{OMA}^{m_sK}\rho^{-m_sK},
\end{split}
\end{equation}
\begin{equation}\label{ch7_eq_Without_4}
\begin{split}
\mathbb{P}_{OMA}^{1,low,\infty}=\xi_2 \tilde{\gamma}_{OMA}^{m_sK}\rho^{-m_sK}.
\end{split}
\end{equation}
Note that \eqref{ch7_eq_Without_4} is also the outage probability for the case of continuous shifts.
\end{proposition}

\begin{proof}
The CDF of $Y$'s lower bound for $y\rightarrow0^+$ has been derived as \eqref{ch7_eq_Appe_Without_4} in the proof of Lemma \ref{ch7_Lemm_Without}. Other steps are similar to the proof of Proposition \ref{ch7_Prop_Without_NOMA}.
\end{proof}

\begin{corollary}\label{ch7_Coro_Without_D_OMA}
In Scenario \uppercase\expandafter{\romannumeral1}, the diversity orders of each UD for both discrete and continuous phase shifts under the OMA scheme are $m_sK$.
\end{corollary}

\begin{remark}
For the OMA case of Scenario \uppercase\expandafter{\romannumeral1}, the diversity order with discrete phase shifts is the same as that with continuous phase shifts.
The diversity order is related to the number of IRS reflecting elements and the Nakagami fading parameters of the BS-IRS-UD link.
\end{remark}

\section{Scenario \uppercase\expandafter{\romannumeral2} (With Direct Link)}\label{ch7_s2}
In this section, we will analyze the outage performance of Scenario \uppercase\expandafter{\romannumeral2}.
To derive the outage probability, the channel statistics are required.
As compared with Scenario \uppercase\expandafter{\romannumeral1}, Scenario \uppercase\expandafter{\romannumeral2} is more complicated to analyze due to the existence of the direct link.
Specifically, the equivalent channel gain is the sum of direct and reflected components that do not follow the same distribution, which increases the difficulty of analysis.

\subsection{IRS Parameters for Scenario \uppercase\expandafter{\romannumeral2}}
When there is a direct link between the BS and each UD, we need to maximize the equivalent channel gain between the BS and U\textsubscript{$n$}, i.e., $|h_n + \mathbf{G}_n \mathbf{\Theta}_n \mathbf{g}_n| =\left|h_n + \sum_{k=1}^{K}\beta_k^n G_k^n g_k^n e^{j\theta_k^n}\right|$.
The optimal discrete phase-shift variables are derived below.
\begin{lemma}
Denote the optimal discrete phase-shift variable by $\hat{\theta}_k^{2,n}$. Then, we have
\begin{equation}\label{ch7_eq_Lemma_Link_With}
\begin{split}
\hat{\theta}_k^{2,n} = \Delta \left( \left\lfloor \frac{\bar{\theta}_k^{2,n}}{\Delta} \right\rfloor+\frac{1}{2}\right), \quad k = 1, 2, \cdots, K,
\end{split}
\end{equation}
where $\bar{\theta}_k^{2,n}=\mathrm{arg}(h_n)-\mathrm{arg}(G_k^n g_k^n)$.
\end{lemma}

\begin{proof}
The optimal continuous phase shifts can be obtained by setting the phases of all $G_k^n g_k^n e^{j\theta_k^n}$ to be the same as $h_n$.
For continuous phase shifts, there is only one solution that is given by $\bar{\theta}_k^{2,n}=\mathrm{arg}(h_n)-\mathrm{arg}(G_k^n g_k^n)$.
After considering discrete phase shifts, we can adjust the phase-shift variable of each element to be close to the optimal continuous phase-shift variable and obtain the optimal discrete phase-shift variable as \eqref{ch7_eq_Lemma_Link_With}. This completes the proof.
\end{proof}

After adopting the optimal continuous phase shifts $\big\{\bar{\theta}_k^{2,n}\big\}$, we have $|h_n + \mathbf{G}_n \mathbf{\Theta}_n \mathbf{g}_n| = |h_n| + \beta \sum_{k=1}^{K}|G_k^n||g_k^n|$.
For discrete phase shifts, the quantization error is given by $\theta_k^{2,n,er}=\hat{\theta}_k^{2,n}-\bar{\theta}_k^{2,n}$, which is uniformly distributed in $\big[-\frac{\Delta}{2}, \frac{\Delta}{2}\big)$.
As such, after adopting the optimal discrete phase-shift variables $\big\{\hat{\theta}_k^{2,n}\big\}$, we can obtain the equivalent channel gain as
\begin{equation}
\begin{split}
|h_n \!+\! \mathbf{G}_n \mathbf{\Theta}_n \mathbf{g}_n| &\!=\! \left| h_n \!+\! \beta \sum_{k=1}^{K} G_k^n g_k^n e^{j\hat{\theta}_k^{2,n}}\right|
\!=\! \left| |h_n|e^{j\mathrm{arg}(h_n)} \!+\! \beta \sum_{k=1}^{K} |G_k^n| |g_k^n| e^{j\mathrm{arg}(G_k^n g_k^n)} e^{j\bar{\theta}_k^{2,n}}e^{j\theta_k^{2,n,er}}\right|\\
&=\left| |h_n| + \beta \sum_{k=1}^{K} |G_k^n| |g_k^n| e^{j\theta_k^{2,n,er}}\right|.
\end{split}
\end{equation}
Due to the introduction of direct link, Scenario \uppercase\expandafter{\romannumeral2} has a better channel quality than Scenario \uppercase\expandafter{\romannumeral1}, which will result in better performance.

\begin{figure}
\begin{center}
\includegraphics[width=0.75\linewidth]{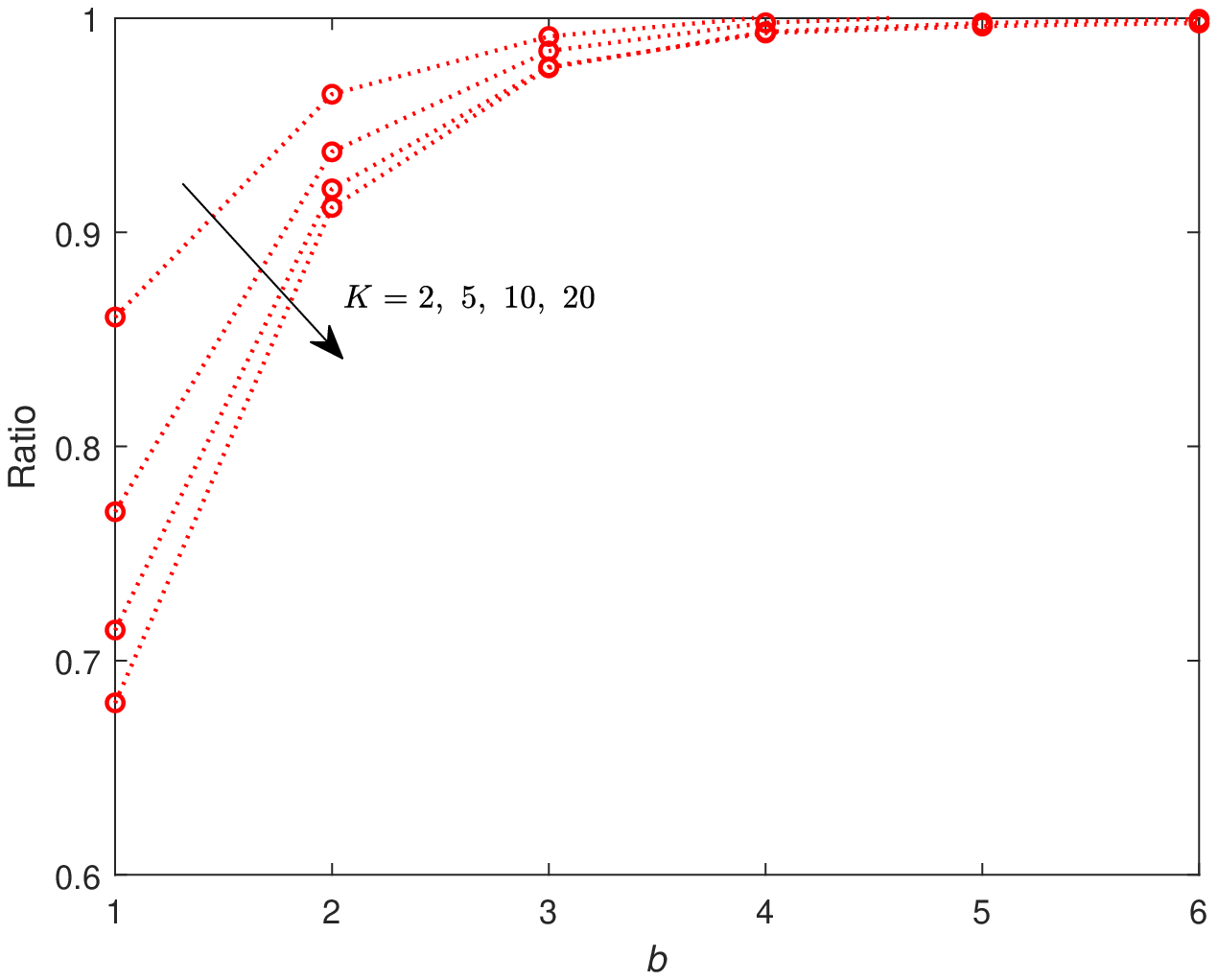}
\end{center}
\caption{The ratio of the average equivalent channel gain with discrete phase shifts to that with continuous phase shifts when $m_G=2$, $m_g=1$, $m_h=1$, and $\beta=0.9$ in Scenario \uppercase\expandafter{\romannumeral2}.}
\label{ch7_Fig_channel_gain_vs_b_with}
\end{figure}

To verify our analysis of the optimal discrete phase shifts, we have conducted Monte Carlo simulations, and the results are shown in Fig. \ref{ch7_Fig_channel_gain_vs_b_with}.
It is observed that the ratio of the equivalent channel gain with discrete phase shifts to that with continuous phase shifts increases as the value of $b$ increases, and gradually approaches $1$.

\subsection{Channel Statistics with Direct Link}
For Scenario \uppercase\expandafter{\romannumeral2}, the new channel statistics for channel gain near $0$ are shown in the following lemma.

\begin{lemma}\label{ch7_Lemm_With}
Denote that $Z_n=\left| |h_n| + \beta \sum_{k=1}^{K} |G_k^n| |g_k^n| e^{j\theta_k^{2,n,er}}\right|$ which is the $n$th channel gain in the ascending order. Based on order statistics, when $m_G\neq m_g$ and $b\ge 2$, the CDF of $Z_n$'s lower bound for $z\rightarrow0^+$ is given by
\begin{equation}\label{ch7_eq_Lemma_With_1}
\begin{split}
F_{Z_n,low}^{0^+}(z) = &\frac{N!}{(N-n)!(n-1)!}\sum_{i=0}^{N-n} \binom{N-n}{i}\frac{(-1)^i}{n+i}
\xi_3^{n+i}z^{(2m_h+2m_sK)(n+i)},
\end{split}
\end{equation}
where $\xi_3=\zeta_2/a^{2m_sK}$ and
$\zeta_2=\frac{(2m_h)^{m_h} \Gamma(2m_h) \big(\sqrt{\pi}4^{m_s-m_l+1}(m_sm_l)^{m_s}\Gamma(2m_s)\Gamma(2m_l-2m_s)\big)^K}
{(2m_h+2m_sK)\Gamma(m_h)\Gamma(2m_h+2m_sK) \big(\Gamma(m_s)\Gamma(m_l)\Gamma\left(m_l-m_s+\frac{1}{2}\right)\big)^K}$.

Note that when $b\rightarrow \infty$, it is the case of continuous phase shifts. The CDF of $Z_n$ for $z\rightarrow0^+$ when $b\rightarrow \infty$ is given by
\begin{equation}\label{ch7_eq_Lemma_With_2}
\begin{split}
F_{Z_n}^{0^+,b\rightarrow \infty}(z) = &\frac{N!}{(N-n)!(n-1)!}\sum_{i=0}^{N-n} \binom{N-n}{i}
\frac{(-1)^i}{n+i}\xi_4^{n+i}z^{(2m_h+2m_sK)(n+i)},
\end{split}
\end{equation}
where $\xi_4=\zeta_2/\beta^{2m_sK}$.
\end{lemma}

\begin{proof}
See Appendix \ref{ch7_Appe_Lemm_With}.
\end{proof}

\subsection{Performance Analysis of NOMA}
In Scenario \uppercase\expandafter{\romannumeral2}, the outage probability of U\textsubscript{$n$} under the NOMA scheme is given by
\begin{equation}
\begin{split}
\mathbb{P}_n^2=1-\mathrm{P_r}(\Psi_{n\leftarrow 1}^2 \ge \tilde{\gamma}_1, \Psi_{n\leftarrow 2}^2 \ge \tilde{\gamma}_2, \cdots, \Psi_{n\leftarrow n}^2 \ge \tilde{\gamma}_n).
\end{split}
\end{equation}
Furthermore, $\mathbb{P}_n^2$ can be transformed into
\begin{equation}
\begin{split}
\mathbb{P}_n^2&=1-\mathrm{P_r}\left(Z_n^2 \ge \tilde{\rho}_{n\leftarrow 1}, Z_n^2 \ge \tilde{\rho}_{n\leftarrow 2}, \cdots, Z_n^2 \ge \tilde{\rho}_{n\leftarrow n}\right)\\
&=\mathrm{P_r}\left(Z_n^2 < \tilde{\rho}_{n,max}\right).
\end{split}
\end{equation}
Based on Lemma \ref{ch7_Lemm_With}, we can characterize the outage performance in the high-SNR regime as shown in the following proposition.

\begin{proposition}
Under the NOMA scheme in Scenario \uppercase\expandafter{\romannumeral2}, when $b\ge 2$, the high-SNR asymptotic expressions for the upper and lower bounds of U\textsubscript{$n$}'s outage probability are as follows:
\begin{equation}\label{ch7_eq_With_1}
\begin{split}
\mathbb{P}_n^{2,up,\infty}& = \frac{N!}{(N-n)!n!}\xi_3^{n}\tilde{\rho}_{n,max}^{n(m_h+m_sK)},
\end{split}
\end{equation}
\begin{equation}\label{ch7_eq_With_3}
\begin{split}
\mathbb{P}_n^{2,low,\infty}&=\frac{N!}{(N-n)!n!}\xi_4^{n} \tilde{\rho}_{n,max}^{n(m_h+m_sK)}.
\end{split}
\end{equation}
Note that \eqref{ch7_eq_With_3} is also the outage probability for the case of continuous shifts.
\end{proposition}

\begin{proof}
For the upper bound, we have
\begin{equation}
\begin{split}
\mathbb{P}_n^{2,up,\infty}=&F_{Z_n,low}^{0^+}\left(\sqrt{\tilde{\rho}_{n,max}}\right)
\approx \frac{N!}{(N-n)!(n-1)!}\\
&\times \sum_{i=0}^{N-n} \binom{N-n}{i}\frac{(-1)^i}{n+i}\xi_3^{n+i}\tilde{\rho}_{n,max}^{(m_h+m_sK)(n+i)}.
\end{split}
\end{equation}
Then, by extracting the lowest-order term, we can derive \eqref{ch7_eq_With_1}.

It is the case of continuous shifts when $b \rightarrow \infty$, and the corresponding outage probability is the lower bound.
We have
\begin{equation}
\begin{split}
\mathbb{P}_n^{2,low,\infty}=&F_{Z_n}^{0^+,b \rightarrow \infty}\left(\sqrt{\tilde{\rho}_{n,max}}\right)
= \frac{N!}{(N-n)!(n-1)!}\\
&\times \sum_{i=0}^{N-n} \binom{N-n}{i}
\frac{(-1)^i}{n+i}\xi_4^{n+i}\tilde{\rho}_{n,max}^{(m_h+m_sK)(n+i)}.
\end{split}
\end{equation}
Next, after extracting the lowest-order term, we can obtain \eqref{ch7_eq_With_3}.
This completes the proof.
\end{proof}

\begin{corollary}\label{ch7_Coro_With_D}
In Scenario \uppercase\expandafter{\romannumeral2}, the diversity orders of the $n$th UD for both discrete and continuous phase shifts under the NOMA scheme are $n(m_h+m_sK)$.
\end{corollary}

\begin{proof}
Since $\tilde{\rho}_{n,max}\propto \rho^{-1}$, the diversity orders can be obtained based on \eqref{ch7_eq_With_1} and \eqref{ch7_eq_With_3}. This completes the proof.
\end{proof}

\begin{remark}
For Scenario \uppercase\expandafter{\romannumeral2}, the proposed system with $b$-bit discrete phase shifts achieves the same diversity order as that with continuous phase shifts when $b\ge2$, i.e., the application of discrete phase shifts does not jeopardize the diversity order.
The diversity order is related to the number of IRS reflecting elements, the Nakagami fading parameters of the direct and BS-IRS-UD links, and the order of the channel gain.
\end{remark}

\subsection{Performance Analysis of OMA}
Since there is no ordering for the channel gains in the OMA scheme, all UDs have the same outage probability.
The outage probability of U\textsubscript{$n$} is given by
\begin{equation}
\begin{split}
\mathbb{P}_{OMA}^2=\mathrm{P_r}(\Psi_{OMA}^2 < \tilde{\gamma}_{OMA}),
\end{split}
\end{equation}
where $\Psi_{OMA}^2=\rho|h_n + \mathbf{G}_n \mathbf{\Theta}_n \mathbf{g}_n|^2$.
Furthermore, $\mathbb{P}_{OMA}^2$ can be transformed into
\begin{equation}
\begin{split}
\mathbb{P}_{OMA}^2&=\mathrm{P_r}\left(Z^2 < \frac{\tilde{\gamma}_{OMA}}{\rho}\right),
\end{split}
\end{equation}
where $Z$ is the unordered equivalent channel gain.

\begin{proposition}
Under the OMA scheme in Scenario \uppercase\expandafter{\romannumeral2}, when $b\ge 2$, the high-SNR asymptotic expressions for the upper and lower bounds of U\textsubscript{$n$}'s outage probability are as follows:
\begin{equation}\label{ch7_eq_With_2}
\begin{split}
\mathbb{P}_{OMA}^{2,up,\infty} = \xi_3 \tilde{\gamma}_{OMA}^{m_h+m_sK}\rho^{-(m_h+m_sK)},
\end{split}
\end{equation}
\begin{equation}\label{ch7_eq_With_4}
\begin{split}
\mathbb{P}_{OMA}^{2,low,\infty} = \xi_4\tilde{\gamma}_{OMA}^{m_h+m_sK}\rho^{-(m_h+m_sK)}.
\end{split}
\end{equation}
Note that \eqref{ch7_eq_With_4} is also the outage probability for the case of continuous shifts.
\end{proposition}

\begin{proof}
The CDF of $Z$ has been derived as \eqref{ch7_eq_Appe_With_3} in the proof of Lemma \ref{ch7_Lemm_With}.
Other steps are similar to the proof of Proposition \ref{ch7_Prop_Without_NOMA}.
\end{proof}

\begin{corollary}\label{ch7_Coro_With_D_OMA}
In Scenario \uppercase\expandafter{\romannumeral2}, the diversity orders of each UD for both discrete and continuous phase shifts under the OMA scheme are $m_h+m_sK$.
\end{corollary}

\begin{remark}
For the OMA case of Scenario \uppercase\expandafter{\romannumeral2}, the diversity order with discrete phase shifts is the same as that with continuous phase shifts.
The diversity order is related to the number of IRS reflecting elements and the Nakagami fading parameters of the direct and BS-IRS-UD links.
\end{remark}

\subsection{Summary of Diversity Orders}
After completing all analyses for both scenarios, all results are summarized in Table \ref{ch7_Tabl_1} for ease of reference.
\begin{table}
\centering
\caption{Diversity orders of U\textsubscript{$n$} (the $n$th UD) for all cases}
\begin{tabular}{|c|c|c|c|c|}
\hline
\centering
MA scheme  & \multicolumn{2}{c|}{NOMA} & \multicolumn{2}{c|}{OMA}  \\
\hline
Phase shifts & Discrete & Continuous & Discrete & Continuous \\
\hline
Scenario \uppercase\expandafter{\romannumeral1}                   & $nm_sK$ & $nm_sK$ & $m_sK$ & $m_sK$ \\
\hline
Scenario \uppercase\expandafter{\romannumeral2}                   & $n(m_h+m_sK)$ & $n(m_h+m_sK)$ & $m_h+m_sK$ & $m_h+m_sK$ \\
\hline
\end{tabular}\label{ch7_Tabl_1}
\end{table}
First, we observe that the discrete phase shifts can achieve the same diversity order as the continuous phase shifts. On the other hand, discrete phase shifts have lower hardware complexity and cost, which is the advantage of discrete phase shifts.
We also notice that the diversity order is improved with the introduction of direct link.
Furthermore, in both scenarios, the lowest diversity order under the NOMA scheme is the same as the diversity order of each UD under the OMA scheme. Hence, NOMA outperforms OMA in terms of diversity order.

\section{Numerical Results}\label{ch7_numerical}

\begin{table}
\centering
\caption{Setting of parameters}
\resizebox{\textwidth}{!}
{
\begin{tabular}{|c|c|}
\hline
Number of IRSs and UDs & $N=2$, $3$, and $4$\\
\hline
Number of IRS reflecting elements & $K=1$, $2$, $3$, $4$, and $5$\\
\hline
Number of resolution bits & $b=2$, $3$, $4$, and $\infty$\\
\hline
Amplitude-reflection coefficient & $\beta=0.9$\\
\hline
Nakagami fading parameters& $m_G=2$, $m_g=1$, and $m_h=1$\\
\hline
Target data rates & $\tilde{R}_i=\tilde{R}=1$ bps/Hz, $i=1,\ 2,\ \cdots,\ N$\\
\hline
\multirow{3}{*}{Power allocation coefficients}   & $N=2$: $\alpha_1=0.9$, $\alpha_2=0.1$\\
                                        & $N=3$: $\alpha_1=0.7$, $\alpha_2=0.2$, and $\alpha_3=0.1$\\
                                        & $N=4$: $\alpha_1=0.6$, $\alpha_2=0.25$, $\alpha_3=0.1$, and $\alpha_4=0.05$\\
\hline
\end{tabular}\label{ch7_Tabl_2}
}
\end{table}

In this section, numerical results are presented for the performance evaluation of the considered network. Monte Carlo simulations are conducted to verify the accuracy.
All $N$ IRSs are pre-deployed so that their service coverages do not overlap, and each IRS selects a UD within its coverage.
The parameters are set by referring to some relevant works \cite{hou2019reconfigurable,cheng2020downlink,ding2014performance}, which are shown in Table \ref{ch7_Tabl_2}.

\subsection{Scenario \uppercase\expandafter{\romannumeral1}}

\begin{figure}[tb]
\begin{center}
\includegraphics[width=0.75\linewidth]{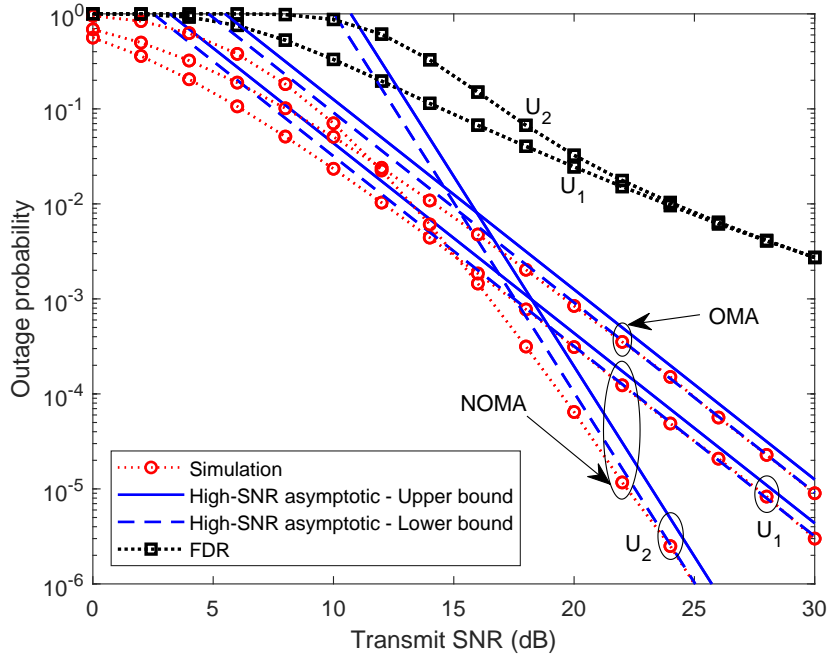}
\end{center}
\caption{Outage probabilities versus SNR in Scenario \uppercase\expandafter{\romannumeral1} when $N=2$, $K=2$, $m_G=2$, $m_g=1$, $\alpha_1=0.9$, $\alpha_2=0.1$, and $b=3$.}
\label{ch7_Fig_OP_Without_N_2_K_2}
\end{figure}

In Fig. \ref{ch7_Fig_OP_Without_N_2_K_2}, two UDs' outage probabilities versus the transmit SNR in Scenario \uppercase\expandafter{\romannumeral1} when $N=2$, $K=2$, and $b=3$ are plotted.
First, we observe that the outage probability of each UD for NOMA gradually approaches the interval between the upper and lower bounds in the high-SNR regime, which are derived from \eqref{ch7_eq_Without_1} and \eqref{ch7_eq_Without_3}, respectively.
Then, by observing the slopes, the diversity orders of U\textsubscript{$1$} and U\textsubscript{$2$} are $2$ and $4$, respectively, which is consistent with Corollary \ref{ch7_Coro_Without_D}.
For OMA, the outage probability points are also located between the upper and lower bounds in the high-SNR regime, which are derived from \eqref{ch7_eq_Without_2} and \eqref{ch7_eq_Without_4}, respectively.
There is only one curve for OMA, as both UDs have the same outage probability statistically.
It is observed that the diversity order for OMA is $2$, which validates Corollary \ref{ch7_Coro_Without_D_OMA}.
This observation demonstrates the superiority of NOMA over OMA from the diversity order perspective.
For comparisons, we regard a multi-FDR-assisted NOMA network as the benchmark.
Specifically, FDRs under the classic protocol are deployed at the places of IRSs to help their respective UDs.
The FDR works under a realistic assumption that is the same as \cite{ding2020simple,zhong2016non}. Specifically, since the relay is aware of its own transmitted symbol, self-interference cancellation can be applied.
We assume that the self-interference cancellation is imperfect and the channel of residual self-interference experiences the Nakagami-$m$ fading.
Since the reflection at the IRS is passive without consuming the energy, for a fair comparison, we assume that the transmit power at the BS and the FDR is $P_r=0.5P$.
We observe that the outage probabilities of both UDs for FDR-NOMA are higher than those for IRS-NOMA.
Moreover, the outage probability of each UD for FDR-NOMA gradually converges to a floor as the transmit SNR increases, due to the self-interference.
These two phenomenons show the disadvantage of FDRs as compared with IRSs.

\begin{figure}[tb]
\begin{center}
\includegraphics[width=0.75\linewidth]{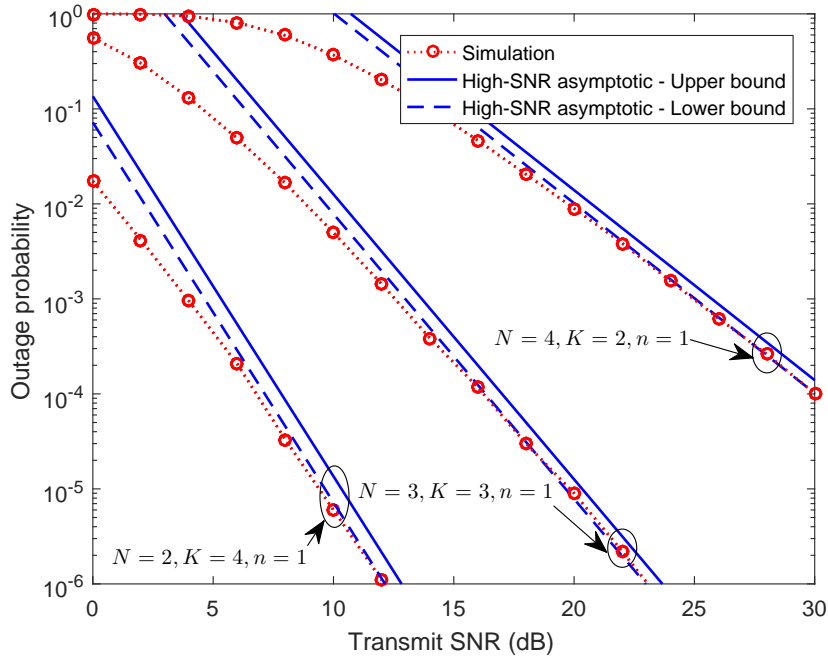}
\end{center}
\caption{Outage probability of U\textsubscript{$1$} versus SNR for NOMA in Scenario \uppercase\expandafter{\romannumeral1} when $m_G=2$, $m_g=1$, and $b=3$.}
\label{ch7_Fig_OP_Without_Various_N_Large_K}
\end{figure}

We further plot the outage probabilities of U\textsubscript{$1$} with different values of $N$ and $K$ in Fig. \ref{ch7_Fig_OP_Without_Various_N_Large_K}.
First, we observe that the outage probability of U\textsubscript{$1$} gradually approaches the interval between the analytical upper and lower bounds in the high-SNR regime.
The observation demonstrates that the diversity orders of U\textsubscript{$1$} are $2$, $3$, and $4$ when $N=4$ $K=2$, $N=3$ $K=3$, and $N=2$ $K=4$, respectively, which further validates our analytical results.

\begin{figure}[tb]
\begin{center}
\includegraphics[width=0.75\linewidth]{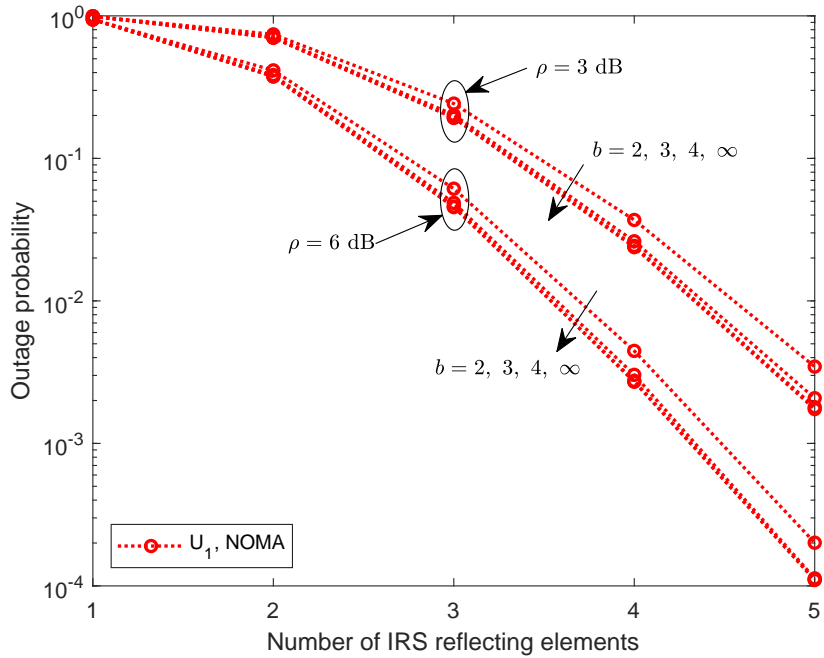}
\end{center}
\caption{Outage probability of U\textsubscript{$1$} versus $K$ for NOMA in Scenario \uppercase\expandafter{\romannumeral1} when $\rho=3$ dB and $6$ dB. Here, $N=3$, $m_G=2$, $m_g=1$, $\alpha_1=0.7$, $\alpha_2=0.2$, and $\alpha_3=0.1$.}
\label{ch7_Fig_OP_K_Without}
\end{figure}

In Fig. \ref{ch7_Fig_OP_K_Without}, the outage probability of U\textsubscript{$1$} versus the number of IRS reflecting elements in Scenario \uppercase\expandafter{\romannumeral1} is plotted.
First, we observe that the outage probability decreases as $K$ increases. Then, we observe that the outage probability for $\rho=6$ dB is lower than that for $\rho=3$ dB. Next, we observe that for each specific transmit SNR, the outage probability decreases as $b$ increases. More importantly, when $b \ge 3$, the outage probability is close to the outage probability of continuous phase shifts. It reveals that a $3$-bit resolution for discrete phase shifts is sufficient to achieve near-optimal performance.

\subsection{Scenario \uppercase\expandafter{\romannumeral2}}

\begin{figure}[tb]
\begin{center}
\includegraphics[width=0.75\linewidth]{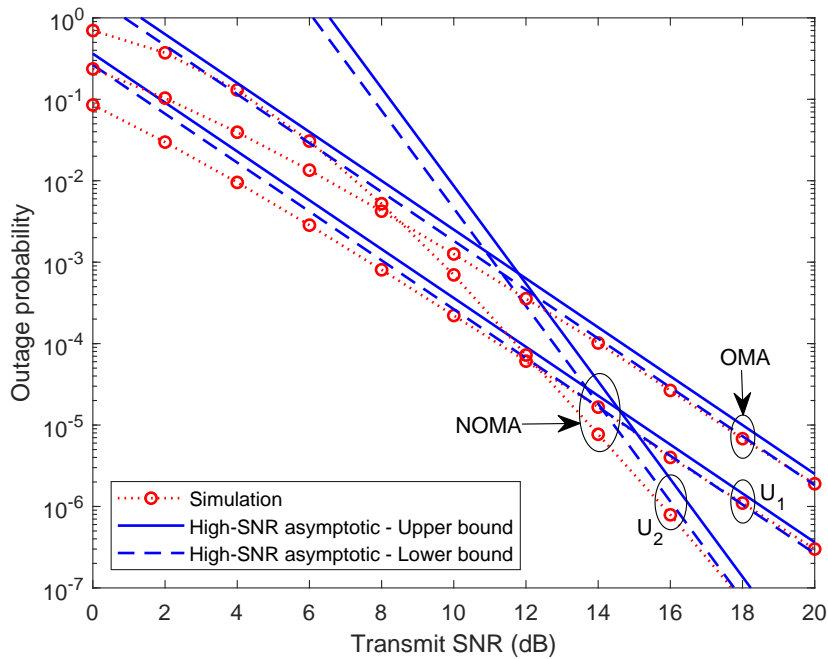}
\end{center}
\caption{Outage probabilities versus SNR in Scenario \uppercase\expandafter{\romannumeral2} when $N=2$, $K=2$, $m_G=2$, $m_g=1$, $m_h=1$, $\alpha_1=0.9$, $\alpha_2=0.1$, and $b=3$.}
\label{ch7_Fig_OP_With_N_2_K_2}
\end{figure}

In Fig. \ref{ch7_Fig_OP_With_N_2_K_2}, two UDs' outage probabilities versus the transmit SNR in Scenario \uppercase\expandafter{\romannumeral2} when $N=2$, $K=2$, and $b=3$ are plotted.
First, we observe that the outage probability of each UD for NOMA gradually approaches the interval between the upper and lower bounds that are derived from \eqref{ch7_eq_With_1} and \eqref{ch7_eq_With_3}, respectively.
Then, by observing the slopes, the diversity orders of U\textsubscript{$1$} and U\textsubscript{$2$} are $3$ and $6$, respectively, which is consistent with Corollary \ref{ch7_Coro_With_D}.
For OMA, the outage probability points are also located between the upper and lower bounds in the high-SNR regime, which are derived from \eqref{ch7_eq_With_2} and \eqref{ch7_eq_With_4}, respectively.
There is only one curve for OMA, as both UDs have the same outage probability statistically.
It is observed that the diversity order for OMA is $3$, which validates Corollary \ref{ch7_Coro_With_D_OMA}.
This observation demonstrates the superiority of NOMA over OMA from the diversity order perspective.

\begin{figure}[tb]
\begin{center}
\includegraphics[width=0.75\linewidth]{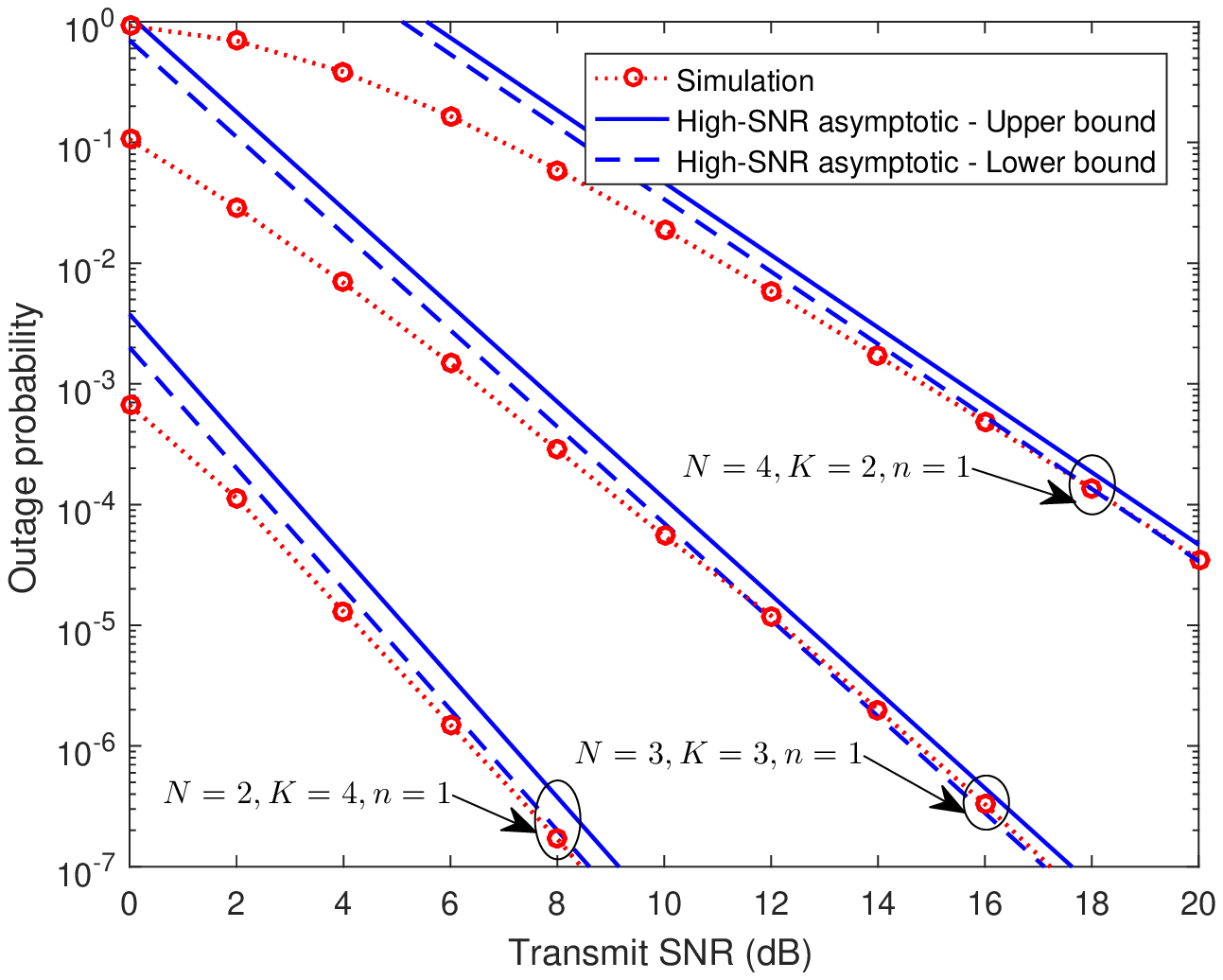}
\end{center}
\caption{Outage probability of U\textsubscript{$1$} versus SNR for NOMA in Scenario \uppercase\expandafter{\romannumeral2} when $m_G=2$, $m_g=1$, $m_h=1$, and $b=3$.}
\label{ch7_Fig_OP_With_Various_N_Large_K}
\end{figure}

We further plot the outage probabilities of U\textsubscript{$1$} with different values of $N$ and $K$ in Fig. \ref{ch7_Fig_OP_With_Various_N_Large_K}.
First, we observe that the outage probability of U\textsubscript{$1$} gradually approaches the interval between the analytical upper and lower bounds in the high-SNR regime.
We observe that the diversity orders of U\textsubscript{$1$} are $3$, $4$, and $5$ when $N=4$ $K=2$, $N=3$ $K=3$, and $N=2$ $K=4$, respectively, which further validates our analysis.
Due to the existence of direct links, it is observed that the scenario with direct link has a larger diversity order than the scenario without direct link.

\begin{figure}[tb]
\begin{center}
\includegraphics[width=0.75\linewidth]{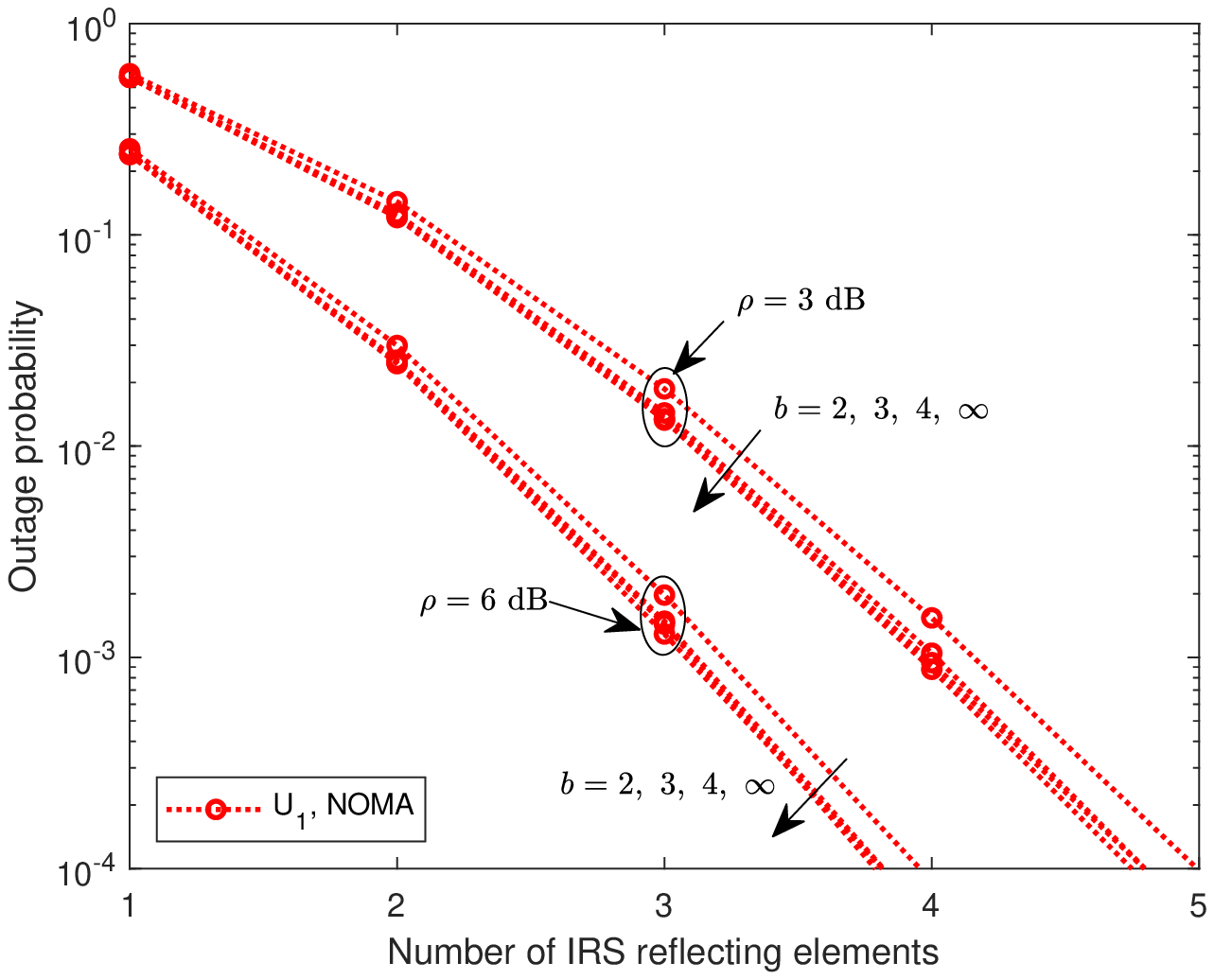}
\end{center}
\caption{Outage probability of U\textsubscript{$1$} versus $K$ for NOMA in Scenario \uppercase\expandafter{\romannumeral2} when $\rho=3$ dB and $6$ dB. Here, $N=3$, $m_G=2$, $m_g=1$, $m_h=1$, $\alpha_1=0.7$, $\alpha_2=0.2$, and $\alpha_3=0.1$.}
\label{ch7_Fig_OP_K_With}
\end{figure}

In Fig. \ref{ch7_Fig_OP_K_With}, the outage probability of U\textsubscript{$1$} versus the number of IRS reflecting elements in Scenario \uppercase\expandafter{\romannumeral2} is plotted.
First, we observe that the outage probability decreases as $K$ increases. Then, we observe that the outage probability for $\rho=6$ dB is lower than that for $\rho=3$ dB. Next, we observe that for each specific transmit SNR, the outage probability decreases with the increase of $b$. Importantly, when $b \ge 3$, the outage probability is close to the outage probability of continuous phase shifts.
It is revealed that when direct links are considered, a $3$-bit resolution is also sufficient for discrete phase shifts to achieve near-optimal performance.

\section{Conclusion}\label{ch7_summary}
We have studied multi-IRS assisted NOMA systems with discrete phase shifts.
We have determined the diversity order by deriving the high-SNR asymptotic expressions for the upper and lower bounds of the outage probability.
It has been demonstrated that the use of discrete phase shift does not degrade diversity order.
Furthermore, it has been shown that a $3$-bit resolution is sufficient for discrete phase shifts to achieve the near-optimal performance.
Therefore, IRSs with discrete phase shifts strike a balance between system performance and hardware cost.
As the CSI is assumed to be perfectly known at the BS in this paper, IRS-assisted networks with imperfect CSI are worth investigating for future work.
In addition, optimizing power allocation for NOMA can further improve the performance of the considered network, which is another promising future research direction.

\begin{appendices}

\section{Proof of Lemma \ref{ch7_Lemm_Without}}\label{ch7_Appe_Lemm_Without}
\renewcommand{\theequation}{\thesection.\arabic{equation}}
\setcounter{equation}{0}

\textit{Unordered channel gain:} First, we need to derive the PDF of the unordered channel gain after adopting the optimal discrete phase shifts $\big\{\hat{\theta}_k^*\big\}$,\footnote{We use the superscript $*$ instead of $n$ to indicate non-ordering in this section.} which is denoted by $Y$.
We have
\begin{equation}\label{ch7_eq_Appen_Without_1}
\begin{split}
Y & = \beta \left|\sum_{k=1}^{K} |G_k^*| |g_k^*| e^{j\theta_k^{1,*,er}}\right|
\ge \beta \sum_{k=1}^{K} |G_k^*| |g_k^*| \cos\left(\theta_k^{1,*,er}\right)
\ge a \sum_{k=1}^{K} |G_k^*| |g_k^*|,
\end{split}
\end{equation}
where the equality holds when $b\rightarrow\infty$.
Next, we denote $Q=\sum_{k=1}^{K}Q_k$, where $Q_k=|G_k^*||g_k^*|$.
Since all $Q_k$ $(k=1, 2, \cdots, K)$ are independent and identically distributed (i.i.d.), according to \cite{bhargav2018product}, the PDF of single $Q_k$ (the product of two Nakagami-$m$ random variables) is given by
\begin{equation}
\begin{split}
f_{Q_k}(q)=\frac{4(m_sm_l)^{\frac{m_s+m_l}{2}}}{\Gamma(m_s)\Gamma(m_l)}q^{m_s+m_l-1}K_{m_s-m_l}(2\sqrt{m_sm_l}q),
\end{split}
\end{equation}
for $q\ge0$, where $K_v(\cdot)$ is the modified Bessel function of the second kind.
The Laplace transform of $f_{Q_k}(q)$ is derived as
\begin{equation}
\begin{split}
\mathcal{L}_{f_{Q_k}}(s)=&\frac{4(m_sm_l)^{\frac{m_s+m_l}{2}}}{\Gamma(m_s)\Gamma(m_l)}
\int_0^{\infty}q^{m_s+m_l-1}
e^{-sq}K_{m_s-m_l}(2\sqrt{m_sm_l}q)dq.
\end{split}
\end{equation}
Furthermore, by referring to \cite[\textrm{eq}. (6.621.3)]{gradshteyn2007}, we have
\begin{equation}
\begin{split}
&\mathcal{L}_{f_{Q_k}}(s)=\phi_1\left(s+2\sqrt{m_sm_l}\right)^{-2m_s}
F\bigg(2m_s,m_s-m_l+\frac{1}{2};m_s+m_l+\frac{1}{2};\frac{s-2\sqrt{m_sm_l}}{s+2\sqrt{m_sm_l}}\bigg),
\end{split}
\end{equation}
where $\phi_1=\frac{\sqrt{\pi}4^{m_s-m_l+1}(m_sm_l)^{m_s}\Gamma(2m_s)\Gamma(2m_l)}{\Gamma(m_s)\Gamma(m_l)\Gamma\left(m_s+m_l+\frac{1}{2}\right)}$ and $F(\cdot, \cdot; \cdot; \cdot)$ is the hypergeometric series.
When $s\rightarrow\infty$, since $m_s<m_l$ satisfies the condition of \cite[\textrm{eq}. (9.122.1)]{gradshteyn2007}, we have
\begin{equation}\label{ch7_eq_Appe_Without_for_With}
\begin{split}
\mathcal{L}_{f_{Q_k}}^{\infty}(s)=\phi_2 s^{-2m_s},
\end{split}
\end{equation}
where $\phi_2=\frac{\sqrt{\pi}4^{m_s-m_l+1}(m_sm_l)^{m_s}\Gamma(2m_s)\Gamma(2m_l-2m_s)}{\Gamma(m_s)\Gamma(m_l)\Gamma\left(m_l-m_s+\frac{1}{2}\right)}$.
Since all $Q_k$ are i.i.d., we have
\begin{equation}\label{ch7_eq_Appe_Without_2}
\begin{split}
\mathcal{L}_{f_Q}^{\infty}(s)=\prod_{k=1}^{K}\mathcal{L}_{f_{Q_k}}^{\infty}(s)=\phi_2^Ks^{-2m_sK}.
\end{split}
\end{equation}
Thus, by performing the inverse Laplace transform for \eqref{ch7_eq_Appe_Without_2} and referring to \cite[\textrm{eq}. (17.13.3)]{gradshteyn2007}, the PDF of $Q$ for $q\rightarrow0^+$ can be derived as
\begin{equation}
\begin{split}
f_Q^{0^+}(q)=\frac{\phi_2^K}{\Gamma(2m_sK)}q^{2m_sK-1}.
\end{split}
\end{equation}
Following that, according to \cite[\textrm{eq}. (3.351.1)]{gradshteyn2007}, the CDF of $Q$ for $q\rightarrow0^+$ can be derived as
\begin{equation}
\begin{split}
F_Q^{0^+}(q)&=\frac{\phi_2^K}{\Gamma(2m_sK)2m_sK}q^{2m_sK}.
\end{split}
\end{equation}
Finally, the PDF and CDF of $Y$'s lower bound for $y\rightarrow0^+$ can be derived as follows:
\begin{equation}
\begin{split}
f_{Y,low}^{0^+}(y)=\frac{1}{a}f_Q^{0^+}\left(\frac{y}{a}\right)=\phi_3 y^{2m_sK-1},
\end{split}
\end{equation}
\begin{equation}\label{ch7_eq_Appe_Without_4}
\begin{split}
F_{Y,low}^{0^+}(y)&=F_Q^{0^+}\left(\frac{y}{a}\right)=\xi_1 y^{2m_sK},
\end{split}
\end{equation}
where $\phi_3=\frac{\phi_2^K}{\Gamma(2m_sK)a^{2m_sK}}$ and $\xi_1 =\frac{\phi_3}{2m_sK}$.

\textit{Ordered channel gain:} Based on order statistics \cite{david2004order}, the CDF of $Y_n$'s lower bound for $y\rightarrow0^+$ is given by
\begin{equation}
\begin{split}
F_{Y_n,low}^{0^+}(y)=&\frac{N!}{(N-n)!(n-1)!}\sum_{i=0}^{N-n} \binom{N-n}{i}
\frac{(-1)^i}{n+i}\left(F_{Y,low}^{0^+}(y)\right)^{n+i}.
\end{split}
\end{equation}
For the case of $b \rightarrow \infty$, the quantization error $\theta_k^{1,n,er} \rightarrow 0$, and the equality holds in \eqref{ch7_eq_Appen_Without_1} with $a=\beta \cos\left(\frac{\pi}{2^b}\right) \rightarrow \beta$.
Therefore, the derived CDF is for $Y_n$, not $Y_n$'s lower bound.
Furthermore, we can obtain \eqref{ch7_eq_Lemma_Without_2} after substituting $a$ by $\beta$ in \eqref{ch7_eq_Lemma_Without_1}.
This completes the proof.

\section{Proof of Lemma \ref{ch7_Lemm_With}}\label{ch7_Appe_Lemm_With}
\renewcommand{\theequation}{\thesection.\arabic{equation}}
\setcounter{equation}{0}

\textit{Unordered channel gain:} First, we need to derive the PDF of the unordered channel gain which is denoted by $Z$.
We have\footnote{We use the subscript/superscript $*$ instead of $n$ to indicate non-ordering in this section.}
\begin{equation}\label{ch7_eq_Appe_With_1}
\begin{split}
Z & \ge |h_*| + \beta \sum_{k=1}^{K} |G_k^*| |g_k^*| \cos\left(\theta_k^{2,*,er}\right)
\ge |h_*| + a \sum_{k=1}^{K} |G_k^*| |g_k^*|,
\end{split}
\end{equation}
where the equality holds when $b\rightarrow\infty$.
Since $h_*$ follows the Nakagami-$m$ fading model, the PDF of $|h_*|$ is given by
\begin{equation}
\begin{split}
f_{|h_*|}(x)=\frac{2m_h^{m_h}}{\Gamma(m_h)}x^{2m_h-1}e^{-m_hx^2}.
\end{split}
\end{equation}
The Laplace transform of $f_{|h_*|}(x)$ is given by
\begin{equation}
\begin{split}
\mathcal{L}_{f_{|h_*|}}(s)=&\frac{2m_h^{m_h}}{\Gamma(m_h)}\int_0^{\infty}x^{2m_h-1}e^{-m_hx^2}e^{-sx}dx.
\end{split}
\end{equation}
Then, by referring to \cite[\textrm{eq}. (3.462.1)]{gradshteyn2007}, we can derive that
\begin{equation}\label{ch7_eq_Appe_With_2}
\begin{split}
\mathcal{L}_{f_{|h_*|}}(s)=&\frac{\Gamma(2m_h)}{\Gamma(m_h)}e^{\frac{s^2}{8m_h}}D_{-2m_h}\left(\frac{s}{\sqrt{2m_h}}\right),
\end{split}
\end{equation}
where $D_p(\cdot)$ is the Parabolic cylinder function.
Furthermore, based on $D_p(s)\approx e^{-\frac{s^2}{4}}s^p$ \cite[\textrm{eq}. (9.246.1)]{gradshteyn2007}, we approximate \eqref{ch7_eq_Appe_With_2} for $s\rightarrow\infty$ as
\begin{equation}
\begin{split}
\mathcal{L}_{f_{|h_*|}}^{\infty}(s)=&\frac{\Gamma(2m_h)(2m_h)^{m_h}}{\Gamma(m_h)}s^{-2m_h}.
\end{split}
\end{equation}
On the other hand, we denote $Q^{'}=\sum_{k=1}^{K}Q_k^{'}$ with $Q_k^{'}=a|G_k^*||g_k^*|=a Q_k$ and have
\begin{equation}
\begin{split}
f_{Q_k^{'}}(q')=\frac{1}{a}f_{Q_k}\left(\frac{q'}{a}\right).
\end{split}
\end{equation}
Then, the Laplace transform of $f_{Q_k^{'}}(q')$ is derived as
\begin{equation}
\begin{split}
\mathcal{L}_{f_{Q_k^{'}}}(s)=\mathcal{L}_{f_{Q_k}}(as).
\end{split}
\end{equation}
Since all $Q_k^{'}$ are i.i.d., we have
\begin{equation}
\begin{split}
\mathcal{L}_{f_{Q^{'}}}(s)=\prod_{k=1}^{K}\mathcal{L}_{f_{Q_k^{'}}}(s)=\left(\mathcal{L}_{f_{Q_k^{'}}}(s)\right)^K.
\end{split}
\end{equation}
Furthermore, since $|h_*|$ and $Q^{'}$ are independent, by referring to \eqref{ch7_eq_Appe_Without_for_With}, we have
\begin{equation}
\begin{split}
\mathcal{L}_{f_{Z,low}}^{\infty}(s)&=\mathcal{L}_{f_{|h_*|}}^{\infty}(s)\mathcal{L}_{f_{Q^{'}}}^{\infty}(s)
=\mathcal{L}_{f_{|h_*|}}^{\infty}(s) \left(\mathcal{L}_{f_{Q_k}}^{\infty}(as)\right)^K \\
&=\frac{\Gamma(2m_h)(2m_h)^{m_h}\phi_2^Ka^{-2m_sK}}{\Gamma(m_h)}s^{-2m_h-2m_sK},
\end{split}
\end{equation}
for $s\rightarrow\infty$.
After carrying out the inverse Laplace transform, we can derive the PDF and CDF of $Z$'s lower bound for $z\rightarrow0^+$ as follows:
\begin{equation}
\begin{split}
f_{Z,low}^{0^+}(z)=\phi_5 z^{2m_h+2m_sK-1},
\end{split}
\end{equation}
\begin{equation}\label{ch7_eq_Appe_With_3}
\begin{split}
F_{Z,low}^{0^+}(z)=\xi_3 z^{2m_h+2m_sK},
\end{split}
\end{equation}
where $\phi_5=\frac{\Gamma(2m_h)(2m_h)^{m_h}\phi_2^Ka^{-2m_sK}}{\Gamma(m_h)\Gamma(2m_h+2m_sK)}$ and $\xi_3=\frac{\phi_5}{2m_h+2m_sK}$.

\textit{Ordered channel gain:}
Based on order statistics \cite{david2004order}, the CDF of $Z_n$'s lower bound for $z\rightarrow0^+$ is given by
\begin{equation}
\begin{split}
F_{Z_n,low}^{0^+}(z)=&\frac{N!}{(N-n)!(n-1)!}\sum_{i=0}^{N-n} \binom{N-n}{i}
\frac{(-1)^i}{n+i}\left(F_{Z,low}^{0^+}(z)\right)^{n+i}.
\end{split}
\end{equation}
For the case of $b \rightarrow \infty$, it is similar to the proof of Lemma \ref{ch7_Lemm_Without}, and we can derive \eqref{ch7_eq_Lemma_With_2}.
This completes the proof.
\end{appendices}

\section*{Acknowledgment}
The authors would like to thank Professor Mohamed-Slim Alouini from King Abdullah University of Science and Technology (KAUST) for helping them improve the system model.

\bibliographystyle{IEEEtran}
\bibliography{ref}

% Generated by IEEEtran.bst, version: 1.14 (2015/08/26)
\begin{thebibliography}{10}
\providecommand{\url}[1]{#1}
\csname url@samestyle\endcsname
\providecommand{\newblock}{\relax}
\providecommand{\bibinfo}[2]{#2}
\providecommand{\BIBentrySTDinterwordspacing}{\spaceskip=0pt\relax}
\providecommand{\BIBentryALTinterwordstretchfactor}{4}
\providecommand{\BIBentryALTinterwordspacing}{\spaceskip=\fontdimen2\font plus
\BIBentryALTinterwordstretchfactor\fontdimen3\font minus
  \fontdimen4\font\relax}
\providecommand{\BIBforeignlanguage}[2]{{%
\expandafter\ifx\csname l@#1\endcsname\relax
\typeout{** WARNING: IEEEtran.bst: No hyphenation pattern has been}%
\typeout{** loaded for the language `#1'. Using the pattern for}%
\typeout{** the default language instead.}%
\else
\language=\csname l@#1\endcsname
\fi
#2}}
\providecommand{\BIBdecl}{\relax}
\BIBdecl

\bibitem{ding2017application}
Z.~Ding, Y.~Liu, J.~Choi, Q.~Sun, M.~Elkashlan, C.-L. I, and H.~V. Poor,
  ``Application of non-orthogonal multiple access in {LTE} and 5{G} networks,''
  \emph{IEEE Commun. Mag.}, vol.~55, no.~2, pp. 185--191, Feb. 2017.

\bibitem{montalban2018multimedia}
J.~Montalban, P.~Scopelliti, M.~Fadda, E.~Iradier, C.~Desogus, P.~Angueira,
  M.~Murroni, and G.~Araniti, ``Multimedia multicast services in 5{G} networks:
  {S}ubgrouping and non-orthogonal multiple access techniques,'' \emph{IEEE
  Commun. Mag.}, vol.~56, no.~3, pp. 91--95, Mar. 2018.

\bibitem{wang2019channel}
Z.~Wang, L.~Liu, and S.~Cui, ``Channel estimation for intelligent reflecting
  surface assisted multiuser communications: {F}ramework, algorithms, and
  analysis,'' \emph{IEEE Trans. Wireless Commun.}, vol.~19, no.~10, pp.
  6607--6620, Oct. 2020.

\bibitem{liu2017non5g}
Y.~Liu, Z.~Qin, M.~Elkashlan, Z.~Ding, A.~Nallanathan, and L.~Hanzo,
  ``Non-orthogonal multiple access for 5{G} and beyond,'' \emph{Proc. IEEE},
  vol. 105, no.~12, pp. 2347--2381, Dec. 2017.

\bibitem{liu2020reconfigurable}
Y.~Liu, X.~Liu, X.~Mu, T.~Hou, J.~Xu, Z.~Qin, M.~Di~Renzo, and N.~Al-Dhahir,
  ``Reconfigurable intelligent surfaces: {P}rinciples and opportunities,''
  \emph{IEEE Commun. Surveys Tuts.}, early access, May 5, 2021, doi:
  10.1109/COMST.2021.3077737.

\bibitem{mu2020exploiting}
X.~Mu, Y.~Liu, L.~Guo, J.~Lin, and N.~Al-Dhahir, ``Exploiting intelligent
  reflecting surfaces in {NOMA} networks: {J}oint beamforming optimization,''
  \emph{IEEE Trans. Wireless Commun.}, vol.~19, no.~10, pp. 6884--6898, Oct.
  2020.

\bibitem{shen2019secrecy}
H.~Shen, W.~Xu, S.~Gong, Z.~He, and C.~Zhao, ``Secrecy rate maximization for
  intelligent reflecting surface assisted multi-antenna communications,''
  \emph{IEEE Commun. Lett.}, vol.~23, no.~9, pp. 1488--1492, Sep. 2019.

\bibitem{zhou2020framework}
G.~Zhou, C.~Pan, H.~Ren, K.~Wang, and A.~Nallanathan, ``A framework of robust
  transmission design for {IRS}-aided {MISO} communications with imperfect
  cascaded channels,'' \emph{IEEE Trans. Signal Process.}, vol.~68, pp.
  5092--5106, 2020.

\bibitem{dong2020secure}
L.~Dong and H.-M. Wang, ``Secure {MIMO} transmission via intelligent reflecting
  surface,'' \emph{IEEE Wireless Commun. Lett.}, vol.~9, no.~6, pp. 787--790,
  Jun. 2020.

\bibitem{feng2020deep}
K.~Feng, Q.~Wang, X.~Li, and C.-K. Wen, ``Deep reinforcement learning based
  intelligent reflecting surface optimization for {MISO} communication
  systems,'' \emph{IEEE Wireless Commun. Lett.}, vol.~9, no.~5, pp. 745--749,
  May 2020.

\bibitem{qingqing2019towards}
Q.~Wu and R.~Zhang, ``Towards smart and reconfigurable environment:
  {I}ntelligent reflecting surface aided wireless network,'' \emph{IEEE Commun.
  Mag.}, vol.~58, no.~1, pp. 106--112, Jan. 2020.

\bibitem{liu2017enhancing}
Y.~Liu, Z.~Qin, M.~Elkashlan, Y.~Gao, and L.~Hanzo, ``Enhancing the physical
  layer security of non-orthogonal multiple access in large-scale networks,''
  \emph{IEEE Trans. Wireless Commun.}, vol.~16, no.~3, pp. 1656--1672, Mar.
  2017.

\bibitem{cui2019secure}
M.~Cui, G.~Zhang, and R.~Zhang, ``Secure wireless communication via intelligent
  reflecting surface,'' \emph{IEEE Wireless Commun. Lett.}, vol.~8, no.~5, pp.
  1410--1414, Oct. 2019.

\bibitem{ding2014performance}
Z.~Ding, Z.~Yang, P.~Fan, and H.~V. Poor, ``On the performance of
  non-orthogonal multiple access in 5{G} systems with randomly deployed
  users,'' \emph{IEEE Signal Process. Lett.}, vol.~21, no.~12, pp. 1501--1505,
  Dec. 2014.

\bibitem{ding2016impact}
Z.~Ding, P.~Fan, and H.~V. Poor, ``Impact of user pairing on 5{G} nonorthogonal
  multiple-access downlink transmissions,'' \emph{IEEE Trans. Veh. Technol.},
  vol.~65, no.~8, pp. 6010--6023, Aug. 2016.

\bibitem{oviedo2017fair}
J.~A. Oviedo and H.~R. Sadjadpour, ``A fair power allocation approach to {NOMA}
  in multi-user {SISO} systems,'' \emph{IEEE Trans. Veh. Technol.}, vol.~66,
  no.~9, pp. 7974--7985, Sep. 2017.

\bibitem{xu2015new}
P.~Xu, Z.~Ding, X.~Dai, and H.~V. Poor, ``A new evaluation criterion for
  non-orthogonal multiple access in 5{G} software defined networks,''
  \emph{IEEE Access}, vol.~3, pp. 1633--1639, Sep. 2015.

\bibitem{liu2016cooperative}
Y.~Liu, Z.~Ding, M.~Elkashlan, and H.~V. Poor, ``Cooperative non-orthogonal
  multiple access with simultaneous wireless information and power transfer,''
  \emph{IEEE J. Sel. Areas Commun.}, vol.~34, no.~4, pp. 938--953, Apr. 2016.

\bibitem{zuo2020resource}
J.~Zuo, Y.~Liu, Z.~Qin, and N.~Al-Dhahir, ``Resource allocation in intelligent
  reflecting surface assisted {NOMA} systems,'' \emph{IEEE Trans. Commun.},
  vol.~68, no.~11, pp. 7170--7183, Nov. 2020.

\bibitem{ozdogan2019intelligent}
{\"O}.~{\"O}zdogan, E.~Bj{\"o}rnson, and E.~G. Larsson, ``Intelligent
  reflecting surfaces: {P}hysics, propagation, and pathloss modeling,''
  \emph{IEEE Wireless Commun. Lett.}, vol.~9, no.~5, pp. 581--585, May 2020.

\bibitem{mishra2019channel}
D.~Mishra and H.~Johansson, ``Channel estimation and low-complexity beamforming
  design for passive intelligent surface assisted {MISO} wireless energy
  transfer,'' in \emph{Proc. IEEE Int. Conf. Acoust., Speech and Signal
  Process. (ICASSP)}, Brighton, UK, May 2019, pp. 4659--4663.

\bibitem{he2019cascaded}
Z.-Q. He and X.~Yuan, ``Cascaded channel estimation for large intelligent
  metasurface assisted massive {MIMO},'' \emph{IEEE Wireless Commun. Lett.},
  vol.~9, no.~2, pp. 210--214, Feb. 2020.

\bibitem{liu2020matrix}
H.~Liu, X.~Yuan, and Y.-J.~A. Zhang, ``Matrix-calibration-based cascaded
  channel estimation for reconfigurable intelligent surface assisted multiuser
  {MIMO},'' \emph{IEEE J. Sel. Areas Commun.}, vol.~38, no.~11, pp. 2621--2636,
  Nov. 2020.

\bibitem{di2020reconfigurable}
M.~Di~Renzo, K.~Ntontin, J.~Song, F.~H. Danufane, X.~Qian, F.~Lazarakis,
  J.~De~Rosny, D.-T. Phan-Huy, O.~Simeone, R.~Zhang \emph{et~al.},
  ``Reconfigurable intelligent surfaces vs. relaying: {D}ifferences,
  similarities, and performance comparison,'' \emph{IEEE Open J. Commun. Soc.},
  vol.~1, pp. 798--807, Jul. 2020.

\bibitem{bjornson2019intelligent}
E.~Bj{\"o}rnson, {\"O}.~{\"O}zdogan, and E.~G. Larsson, ``Intelligent
  reflecting surface versus decode-and-forward: {H}ow large surfaces are needed
  to beat relaying?'' \emph{IEEE Wireless Commun. Lett.}, vol.~9, no.~2, pp.
  244--248, Feb. 2020.

\bibitem{lyu2020spatial}
J.~Lyu and R.~Zhang, ``Spatial throughput characterization for intelligent
  reflecting surface aided multiuser system,'' \emph{IEEE Wireless Commun.
  Lett.}, vol.~9, no.~6, pp. 834--838, Jun. 2020.

\bibitem{zhang2019analysis}
\BIBentryALTinterwordspacing
Z.~Zhang, Y.~Cui, F.~Yang, and L.~Ding, ``Analysis and optimization of outage
  probability in multi-intelligent reflecting surface-assisted systems,'' 2019,
  \textit{arXiv:1909.02193}. [Online]. Available:
  \url{https://arxiv.org/abs/1909.02193}
\BIBentrySTDinterwordspacing

\bibitem{yu2019miso}
X.~Yu, D.~Xu, and R.~Schober, ``{MISO} wireless communication systems via
  intelligent reflecting surfaces,'' in \emph{Proc. IEEE/CIC Int. Conf. Commun.
  China (ICCC)}, Changchun, China, Aug. 2019, pp. 735--740.

\bibitem{wu2019intelligent}
Q.~Wu and R.~Zhang, ``Intelligent reflecting surface enhanced wireless network
  via joint active and passive beamforming,'' \emph{IEEE Trans. Wireless
  Commun.}, vol.~18, no.~11, pp. 5394--5409, Nov. 2019.

\bibitem{han2020intelligent}
H.~Han, J.~Zhao, D.~Niyato, M.~Di~Renzo, and Q.-V. Pham, ``Intelligent
  reflecting surface aided network: {P}ower control for physical-layer
  broadcasting,'' in \emph{Proc. IEEE Int. Conf. Commun. (ICC)}, Dublin,
  Ireland, Jun. 2020, pp. 1--7.

\bibitem{guan2020intelligent}
X.~Guan, Q.~Wu, and R.~Zhang, ``Intelligent reflecting surface assisted secrecy
  communication: {I}s artificial noise helpful or not?'' \emph{IEEE Wireless
  Commun. Lett.}, vol.~9, no.~6, pp. 778--782, Jun. 2020.

\bibitem{xu2019discrete}
J.~Xu, W.~Xu, and A.~L. Swindlehurst, ``Discrete phase shift design for
  practical large intelligent surface communication,'' in \emph{Proc. IEEE
  Pacific Rim Conf. Commun., Comput. and Signal Process. (PACRIM)}, Victoria,
  BC, Canada, Aug. 2019, pp. 1--5.

\bibitem{you2020channel}
C.~You, B.~Zheng, and R.~Zhang, ``Channel estimation and passive beamforming
  for intelligent reflecting surface: {D}iscrete phase shift and progressive
  refinement,'' \emph{IEEE J. Sel. Areas Commun.}, vol.~38, no.~11, pp.
  2604--2620, Nov. 2020.

\bibitem{guo2019weighted}
H.~Guo, Y.-C. Liang, J.~Chen, and E.~G. Larsson, ``Weighted sum-rate
  maximization for intelligent reflecting surface enhanced wireless networks,''
  in \emph{Proc. IEEE Global Commun. Conf. (GLOBECOM)}, Waikoloa, HI, USA, Dec.
  2019, pp. 1--6.

\bibitem{abeywickrama2020intelligent}
S.~Abeywickrama, R.~Zhang, Q.~Wu, and C.~Yuen, ``Intelligent reflecting
  surface: {P}ractical phase shift model and beamforming optimization,''
  \emph{IEEE Trans. Commun.}, vol.~68, no.~9, pp. 5849--5863, Sep. 2020.

\bibitem{wu2019beamformingICASSP}
Q.~Wu and R.~Zhang, ``Beamforming optimization for intelligent reflecting
  surface with discrete phase shifts,'' in \emph{Proc. IEEE Int. Conf. Acoust.,
  Speech and Signal Process. (ICASSP)}, Brighton, UK, May 2019, pp. 7830--7833.

\bibitem{wu2019beamforming}
------, ``Beamforming optimization for wireless network aided by intelligent
  reflecting surface with discrete phase shifts,'' \emph{IEEE Trans. Commun.},
  vol.~68, no.~3, pp. 1838--1851, Mar. 2019.

\bibitem{ye2020joint}
J.~Ye, S.~Guo, and M.-S. Alouini, ``Joint reflecting and precoding designs for
  {SER} minimization in reconfigurable intelligent surfaces assisted {MIMO}
  systems,'' \emph{IEEE Trans. Wireless Commun.}, vol.~19, no.~8, pp.
  5561--5574, Aug. 2020.

\bibitem{ding2020simple}
Z.~Ding and H.~V. Poor, ``A simple design of {IRS-NOMA} transmission,''
  \emph{IEEE Commun. Lett.}, vol.~24, no.~5, pp. 1119--1123, May 2020.

\bibitem{ding2020impact}
Z.~Ding, R.~Schober, and H.~V. Poor, ``On the impact of phase shifting designs
  on {IRS-NOMA},'' \emph{IEEE Wireless Commun. Lett.}, vol.~9, no.~10, pp.
  1596--1600, Oct. 2020.

\bibitem{zhu2019power}
J.~Zhu, Y.~Huang, J.~Wang, K.~Navaie, and Z.~Ding, ``Power efficient
  {IRS}-assisted {NOMA},'' \emph{IEEE Trans. Commun.}, vol.~69, no.~2, pp.
  900--913, Feb. 2021.

\bibitem{fu2019intelligent}
M.~Fu, Y.~Zhou, and Y.~Shi, ``Intelligent reflecting surface for downlink
  non-orthogonal multiple access networks,'' in \emph{Proc. IEEE Global Commun.
  Conf. (GLOBECOM)}, Waikoloa, HI, USA, Dec. 2019, pp. 1--6.

\bibitem{yang2020intelligent}
G.~Yang, X.~Xu, and Y.-C. Liang, ``Intelligent reflecting surface assisted
  non-orthogonal multiple access,'' in \emph{Proc. IEEE Wireless Commun. and
  Netw. Conf. (WCNC)}, Seoul, South Korea, May 2020, pp. 1--6.

\bibitem{liu2020ris}
X.~Liu, Y.~Liu, Y.~Chen, and H.~V. Poor, ``{RIS} enhanced massive
  non-orthogonal multiple access networks: {D}eployment and passive beamforming
  design,'' \emph{J. Sel. Areas Commun.}, vol.~39, no.~4, pp. 1057--1071, Apr.
  2021.

\bibitem{hou2019reconfigurable}
T.~Hou, Y.~Liu, Z.~Song, X.~Sun, Y.~Chen, and L.~Hanzo, ``Reconfigurable
  intelligent surface aided {NOMA} networks,'' \emph{IEEE J. Sel. Areas
  Commun.}, vol.~38, no.~11, pp. 2575--2588, Nov. 2020.

\bibitem{cheng2020downlink}
Y.~Cheng, K.~H. Li, Y.~Liu, K.~C. Teh, and H.~V. Poor, ``Downlink and uplink
  intelligent reflecting surface aided networks: {NOMA} and {OMA},'' \emph{IEEE
  Trans. Wireless Commun.}, early access, Feb. 2, 2021, doi:
  10.1109/TWC.2021.3054841.

\bibitem{gao2020distributed}
Y.~Gao, J.~Xu, W.~Xu, D.~W.~K. Ng, and M.-S. Alouini, ``Distributed {IRS} with
  statistical passive beamforming for {MISO} communications,'' \emph{IEEE
  Wireless Commun. Lett.}, vol.~10, no.~2, pp. 221 -- 225, Feb. 2021.

\bibitem{zhong2016non}
C.~Zhong and Z.~Zhang, ``Non-orthogonal multiple access with cooperative
  full-duplex relaying,'' \emph{IEEE Commun. Lett.}, vol.~20, no.~12, pp.
  2478--2481, Dec. 2016.

\bibitem{bhargav2018product}
N.~Bhargav, C.~R.~N. da~Silva, Y.~J. Chun, {\'E}.~J. Leonardo, S.~L. Cotton,
  and M.~D. Yacoub, ``On the product of two $\kappa$-$\mu$ random variables and
  its application to double and composite fading channels,'' \emph{IEEE Trans.
  Wireless Commun.}, vol.~17, no.~4, pp. 2457--2470, Apr. 2018.

\bibitem{gradshteyn2007}
I.~S. Gradshteyn and I.~M. Ryzhik, \emph{Table of Integrals, Series, and
  Products}, 7th~ed.\hskip 1em plus 0.5em minus 0.4em\relax Amsterdam, The
  Netherlands: Elsevier, 2007.

\bibitem{david2004order}
H.~A. David and H.~N. Nagaraja, \emph{Order Statistics}, 3rd~ed.\hskip 1em plus
  0.5em minus 0.4em\relax Hoboken, NJ, USA: Wiley, 2004.

\end{thebibliography}

\end{document}